\documentclass[USenglish,oneside,twocolumn]{article}

\usepackage[utf8]{inputenc}
\usepackage[big]{dgruyter_NEW}
 
\cclogo{\includegraphics{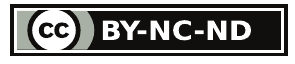}}

\usepackage[noend]{algpseudocode}
\usepackage{amsmath}
\usepackage{amsthm}
\usepackage{amssymb}
\usepackage{array}
\usepackage{booktabs}
\usepackage{breakurl}
\usepackage{caption}
\usepackage[capitalize]{cleveref}
\usepackage{color}
\usepackage{comment}
\usepackage{enumitem}
\usepackage{etoolbox}
\usepackage{float}
\usepackage{graphicx}
\usepackage{hyperref}
\usepackage{mathtools}
\usepackage{pgfplots}
\usepackage{url}
\usepackage{tabularx}
\usepackage{threeparttable} 
\usepackage{tikz}
\usepackage{units}
\usepackage{verbatim}
\usepackage{wrapfig}
\usepackage{xspace}

\usepackage{comment}
\excludecomment{confversion}\includecomment{fullversion}
\newcommand{\confversioncmd}[1]{}
\begin{confversion}
  \renewcommand{\confversioncmd}[1]{#1}
\end{confversion}
\newcommand{\fullversioncmd}[1]{}
\begin{fullversion}
  \renewcommand{\fullversioncmd}[1]{#1}
\end{fullversion}
\excludecomment{proofsketchesinbody}

\newcommand{\commentmarker}[1]{\colorbox{red}{\textcolor{white}{#1}}}
\newcounter{CommentCounter}
\newcommand{\commentUp}[3][]{
  \stepcounter{CommentCounter}
  {\scriptsize\commentmarker{\theCommentCounter}}
  \marginpar{
    \vspace{-#2}
    \tiny\raggedright
    \commentmarker{\theCommentCounter}
    \ifstrempty{#1}{}{\textcolor{red}{#1: }}
    {\scriptsize #3}
  }
}
\renewcommand{\comment}[2][]{\commentUp[#1]{0pt}{#2}}
\renewcommand{\comment}[2][]{}

\newcommand{\name}{\textsc{Tandem}\xspace}

\newcommand{\para}[1]{\vspace{1mm}\noindent\textbf{#1}} 
\newcommand{\parait}[1]{\noindent\textit{#1}} 

\setlist[description]{leftmargin=\parindent}
\setlist[enumerate,1]{leftmargin=17pt}

\theoremstyle{theorem}
\newtheorem{lemma}{Lemma}
\newtheorem{theorem}{Theorem}
\theoremstyle{definition}
\newtheorem{definition}{Definition}
\theoremstyle{tandemdef}
\newtheorem{property}{Property}
\newtheorem{game}{Game}[section]
\newtheorem{protocol}{Protocol}

\newtheoremstyle{tandemdef}
{}
{}
{}
{}
{\scshape}
{.}
{ }
{}

\usetikzlibrary{shapes}
\usetikzlibrary{patterns}
\usetikzlibrary{calc,positioning,shapes.geometric,patterns,intersections,decorations.pathreplacing}
\tikzstyle{mathbox} = [inner sep=0pt, anchor=base,
node distance=0pt, column sep=0em, thick]

\newcommand{\code}[1]{\textsf{#1}\xspace}

\newcommand{\randin}{\in_{R}}

\newcommand{\generator}{g}
\newcommand{\G}{\mathbb{G}}
\newcommand{\grouporder}{p}
\newcommand{\Zp}{\mathbb{Z}_{\grouporder}}

\newcommand{\biggenerator}{\overline{\generator}}
\newcommand{\biggeneratorh}{\overline{h}}
\newcommand{\bigG}{\overline{\G}}
\newcommand{\biggrouporder}{\overline{\grouporder}}
\newcommand{\bigZp}{\mathbb{Z}_{\biggrouporder}}

\newcommand{\schnorrpk}{h}

\newcommand{\tikzmagicarrow}[3]{
\begin{tikzpicture}[]
  \draw[#1](0,0) -- node[above=-0.5ex]{\ensuremath{#2}} (#3,0);
  \node[draw=none] (bottom) at (0,-0.5ex) {};
  \node[draw=none] (top) at (0,1ex) {};
  \node[draw=none] (lowerleft) at (bottom-|current bounding box.west) {};
  \node[draw=none] (topright) at (top-|current bounding box.east) {};
  \pgfresetboundingbox
  \draw[draw=none,use as bounding box] (lowerleft) rectangle (topright);
\end{tikzpicture}
}
\newcommand{\tikzlongarrow}[2]{\tikzmagicarrow{#1}{#2}{1.5}}
\newcommand{\tikzshortarrow}[2]{\tikzmagicarrow{#1}{#2}{0.65}}
\newcommand{\diagramsend}[1]{\tikzlongarrow{->}{#1}}
\newcommand{\diagramrecv}[1]{\tikzlongarrow{<-}{#1}}
\newcommand{\diagramshortsend}[1]{\tikzshortarrow{->}{#1}}
\newcommand{\diagramshortrecv}[1]{\tikzshortarrow{<-}{#1}}

\newcommand{\PaillierN}{N}

\newcommand{\secret}{\mathit{x}}
\newcommand{\xp}{\mathit{x_U}}
\newcommand{\xs}{\mathit{x_S}}

\newcommand{\xpfresh}{\mathit{\tilde{x}_U}}
\newcommand{\xsfresh}{\mathit{\tilde{x}_S}}

\newcommand{\chal}{\mathit{c}}

\newcommand{\phresp}{\mathit{r_U}}

\newcommand{\PK}{\textsf{PK}}

\newcommand{\Enc}{\mathsf{Enc}}

\newcommand{\?}{\stackrel{?}{=}}

\newcommand{\User}{U}

\newcommand{\Ucommit}{\tilde{U}}
\newcommand{\serrand}{t_S}
\newcommand{\phrand}{t_U}
\newcommand{\srand}{\hat{s}'}
\newcommand{\phcommit}{u_U}
\newcommand{\sercommit}{u_S}
\newcommand{\serResp}{r_S}
\newcommand{\phResp}{r_U}
\newcommand{\sresp}{s_{s'}}
\newcommand{\xusresp}{r}

\newcommand{\epoch}{\epsilon}

\newcommand{\pkeygen}{\code{HE.Keygen}^+}
\newcommand{\penc}[1][\ppk]{\mathbf{E}^{+}_{#1}}
\newcommand{\pdec}[1][\psk]{\mathbf{D}^{+}_{#1}}
\newcommand{\prand}{\kappa}
\newcommand{\ppk}[1][]{pk_{#1}}
\newcommand{\psk}[1][]{sk_{#1}}

\newcommand{\prandspace}{\mathcal{R}}

\newcommand{\ctxt}{c}

\newcommand{\pailliermodulus}{N}

\newcommand{\tokensignature}{\sigma}
\newcommand{\ksstoken}{\tau}

\newcommand{\ksstokencontentWithRange}[2][i]{(\bssignature,\allowbreak
  \kssfreshcommitment,\allowbreak
  \kssfreshcommitmentrand,\allowbreak \epoch, \ctxt,
  \xsencdelta, \prand, \allowbreak (\ctxt_{#1},\allowbreak \prand_{#1},\allowbreak \kssrandomizer_{#1})_{#2})}
\newcommand{\ksstokencontent}[1][i]{\ksstokencontentWithRange[#1]{#1=1,\ldots,\tokensecpar}}

\newcommand{\lengthdelta}{\ell_{\kssrandomizer}}
\newcommand{\lengthdeltaval}{\lceil \log \grouporder \rceil + \secpar + \log(\tokensecpar + 1) + 2}
\newcommand{\tokensecpar}{k}
\newcommand{\xsencdelta}{\delta}
\newcommand{\kssrandomizer}{\mu}
\newcommand{\ksscommitment}{\bscommitment}
\newcommand{\ksscommitmentrand}{\bscommitmentrand}
\newcommand{\kssfreshcommitment}{\bsfreshcommitment}
\newcommand{\kssfreshcommitmentrand}{\bsfreshcommitmentrand}
\newcommand{\deltacommitment}{\Delta}
\newcommand{\deltacommitmentrand}{\xi}

\newcommand{\tokendisclose}{\mathcal{D}}
\newcommand{\tokenhidden}{\mathcal{H}}

\newcommand{\tokenid}{id}
\newcommand{\enctokenid}{\overline{\tokenid}}

\newcommand{\commit}{\code{Commit}}
\newcommand{\extcommit}[2]{\code{ExtCommit}(#1, #2)}

\newcommand{\sksym}{x}
\newcommand{\sku}[1]{\sksym_{\ifstrempty{#1}{}{#1,}U}}
\newcommand{\sks}[1]{\sksym_{\ifstrempty{#1}{}{#1,}S}}
\newcommand{\xsenc}[1][]{\overline{\sks{#1}}}

\newcommand{\lastuser}{n}

\newcommand{\hash}{H}
\newcommand{\secpar}{\ell}

\newcommand{\TS}{TS\xspace}
\newcommand{\TSs}{TSs\xspace}

\newcommand{\discloseSetCommitment}{\Delta}
\newcommand{\discloseSetCommitmentRand}{\theta}

\newcommand{\kssrandomizerReveal}{\gamma}
\newcommand{\ksscommitmentrandReveal}{\rho}
\newcommand{\prandReveal}{\nu}

\newcommand{\kssrandomizerDiff}{\gamma}

\newcommand{\Adv}{\mathcal{A}}

\newcommand{\U}{U}

\newcommand{\query}[1]{\code{#1}\xspace}

\newcommand{\qRunTCP}{\query{RunTCP}}

\newcommand{\AdvB}{\ensuremath{\mathcal{B}}\xspace}
\newcommand{\challengebit}{b}

\newcommand{\tandemsetup}{\code{Setup}}
\newcommand{\registeruser}{\code{Register\-User}}
\newcommand{\obtainkstoken}{\code{Obtain\-KeyShare\-Token}}
\newcommand{\genshares}{\code{GenShares}}
\newcommand{\blockshare}{\code{BlockShare}}

\newcommand{\qTCP}{\code{TCP}}

\newcommand{\tcp}{\code{TCP}}
\newcommand{\tcprp}{\code{P}}
\newcommand{\tcpu}{\code{TCP.U}}
\newcommand{\tcpts}{\code{TCP.TS}}
\newcommand{\auxsym}{\textsf{in}}
\newcommand{\auxrp}{\auxsym_{SP}}
\newcommand{\auxu}{\auxsym_{U}}
\newcommand{\TCP}{\tcp}

\newcommand{\NatNumUpTo}[1]{[#1]}
\newcommand{\randspace}{\mathcal{R}}

\newcommand{\skI}{sk_I}
\newcommand{\pkI}{pk_I}

\newcommand{\treg}{r}
\newcommand{\tobtain}[1]{o_{#1}}
\newcommand{\tspend}[1]{s_{#1}}
\newcommand{\tblock}{b}

\newcommand{\blindsigname}{BSA}
\newcommand{\bssetup}{\code{\blindsigname.KeyGen}}
\newcommand{\bssign}{\code{\blindsigname.BlindSign}}
\newcommand{\bsverify}{\code{\blindsigname.Verify}}
\newcommand{\bscommitment}{C}
\newcommand{\bscommitmentrand}{r}
\newcommand{\bsfreshcommitment}{\tilde{C}}
\newcommand{\bsfreshcommitmentrand}{\tilde{r}}
\newcommand{\bssignature}{\sigma}
\newcommand{\bspk}{pk_{\sigma}}
\newcommand{\bssk}{sk_{\sigma}}

\newcommand{\abcpku}{pk_U}
\newcommand{\abcsku}{sk_U}

\newcommand{\cred}{\textsf{cred}}

\newcommand{\revocationlist}{\code{rlist}}

\newcommand{\nrattributes}{k} 
\begin{document}

\author*[1]{Wouter Lueks}
\author[2]{Brinda Hampiholi}
\author[3]{Greg Alp\'ar}
\author[4]{Carmela Troncoso}

\affil[1]{SPRING Lab, EPFL; E-mail: wouter.lueks@epfl.ch}
\affil[2]{Philips Research, all work done while a PhD student at Radboud University; E-mail: brinda.hampiholi@philips.com}
\affil[3]{Open University of the Netherlands, and Radboud University; E-mail: greg.alpar@ou.nl}
\affil[4]{SPRING Lab, EPFL; E-mail: carmela.troncoso@epfl.ch}
  
\title{\huge \name: Securing Keys by Using a Central Server While Preserving Privacy}
\runningtitle{\name: Securing Keys by Using a Central Server While Preserving Privacy}

\begin{abstract}
{
Users' devices, e.g., smartphones or laptops, are typically incapable of securely storing and processing cryptographic keys.
We present \name, a novel set of protocols for securing cryptographic keys
with support from a central server.
\name uses \emph{one-time-use key-share tokens} to preserve users' privacy
with respect to a malicious central server. Additionally, \name enables users to block their keys if they lose their device, and it enables the server to limit how often an adversary can use an unblocked key. We prove \name's security and privacy properties, apply \name to attribute-based credentials, and implement a \name proof of concept to show that it causes little overhead.
}
\end{abstract}
\keywords{privacy, threshold cryptography, anonymity}

\journalname{Proceedings on Privacy Enhancing Technologies}
\DOI{Editor to enter DOI}
\startpage{1}
\received{..}
\revised{..}
\accepted{..}

\journalyear{..}
\journalvolume{..}
\journalissue{..}

\maketitle

\section{Introduction}

Privacy-preserving technologies are increasingly deployed in real-life
scenarios, such as identity management, e.g.,
Kryptic\footnote{\url{https://kryptik.com/}} and IRMA~\cite{IRMAhotpets2017}, and
digital currencies~\cite{MiersG0R13} and products, e.g., the Vega Protocol\footnote{\url{https://vegaprotocol.io/}}. 
These technologies rely on cryptographic protocols and thus their security and privacy properties hinge on the security of the underlying cryptographic keys.

Privacy-preserving cryptographic protocols are typically executed on users' devices
such as phones, tablets, and laptops. These devices, however, are notoriously
hard to secure and the cryptographic keys are often stored or processed insecurely.
A common approach to increase the security of the users' keys is to rely on 
threshold cryptography to not store the entire key in one device. 
Instead, these protocols secret-share the user's key between two or more parties, 
and enable the key's use without ever reconstructing it.
This prevents passive and active attackers from learning the key.
Typical solutions, such as Shatter~\cite{atwater2016shatter}, assume the user has multiple trustworthy devices that can hold shares. 
While effective, this approach has availability and usability issues. To address
these issues, some approaches instead
rely on a highly available central server to store the shares that are not held in the user's device.

Na\"ively using threshold cryptography with a central server, however, can harm
the users' privacy. As the central server is involved in every key use, 
it learns the users' key-usage patterns. This information can enable deanonymization of anonymous transactions by correlating 
key uses with public activities; e.g., by correlating key usages with updates to a blockchain ledger~\cite{JawaheriSBE18,GoldfederKRN18}.

We present \name, a set of protocols that augment threshold-cryptographic
schemes in order to enable
the use of a central server as share holder for increased security, while \emph{preserving} the privacy of key-usage patterns. 
In \name, users obtain \emph{one-time-use key-share tokens} from the central server which contain 
randomized versions of the central server's key share. 
To use their key, users send a key-share token to the server via an anonymous communication channel.
The embedded randomized key share enables the server to run the threshold-cryptographic protocol \emph{without} learning the user's identity.

The construction of key-share tokens decouples the obtaining and using of
tokens.
The one-time property provides two security features: the blocking of keys and
the rate limiting of unblocked keys. These hold even against active
attackers that can compromise the user's device, and thus obtain the user's key share and
key-share tokens, and can observe usage of the key. 

\name can be used to secure the keys of any cryptographic scheme (e.g., encryption, signature, or payments) 
for which a linearly randomizable threshold-cryptographic version exists; this version  
does not require information that identifies the user besides the key; and this
version is secure when operating on linearly randomized keys. 
Not all threshold-cryptographic schemes satisfy these properties. 
For example, \name cannot be applied to existing threshold DSA schemes because either they are 
multiplicative~\cite{MacKenzieR04} or they require identifying information~\cite{GennaroGN16,GennaroG18}.

To demonstrate the potential of \name, we apply it to a threshold version of BBS+ attribute-based credentials (ABCs). ABCs~\cite{BBS+AuSM06,brands2000rethinking,Idemix1-CamenischH02,Idemix2-CamenischL02} protect users' anonymity during authentication.
Introducing a central server that could learn which services users access defeats the very purpose of ABCs. 
\name enables users to securely use their credential keys \emph{without} the central server learning who is using the key, 
preserving the user's privacy even if the central server and the service provider collude. 
\name can be used to extend many privacy-enhancing technologies such as other attribute-based credential 
schemes~\cite{brands2000rethinking,Idemix1-CamenischH02,Idemix2-CamenischL02}, group signatures~\cite{BichselCNSW10,CamenischL04}, 
and electronic cash schemes~\cite{CamenischHL05,MiersG0R13}.

\name can also be beneficial to non-privacy-preserving cryptographic primitives.
For instance, it can be used with threshold variants of Schnorr signatures~\cite{GennaroJKR07} or ElGamal-based encryption schemes~\cite{ElGamal84,ShoupG02} 
to increase key security without revealing usage patterns to the third party storing the key shares.

We evaluate the practicality of the \name protocols on a prototype C implementation. When using the key, \name introduces a 50\,--\,100\,ms overhead on the central server with respect to traditional threshold-cryptographic solutions. For the user, it only adds 5--25\,ms overhead. This is negligible with respect to the delay introduced by the use of
anonymous communications, which is needed in privacy-preserving use cases.

In summary, we make the following contributions:

\sloppypar\para{\checkmark} We formalize the security and privacy properties required when using a threshold-cryptographic protocol with a central server.

\sloppypar\para{\checkmark} We introduce \name. It enables the use of threshold-cryptographic protocols with a central server without revealing key-usage patterns. \name also enables blocking and rate limiting of key usage against active attackers.

\para{\checkmark} We provide a threshold version of BBS+ attribute-based credentials~\cite{BBS+AuSM06}, and show how \name augments its security while retaining its privacy properties.

\para{\checkmark} We prove the security and privacy of \name, and validate its practicality on a prototype implementation. The time-critical computations take less than 100\,ms, imposing reasonable overhead on server and users.

 \section{Related Work}
\label{sec:related-work}

Existing solutions to protect cryptographic keys fall into two categories: \textit{single-party} and \textit{decentralized}. Securing single-party software-based solutions is hard~\cite{Hern15,VeenFLGMVBRG16,WIREDinsecurelaptops,Lipp2018meltdown}. Secure hardware~\cite{ekberg2014untapped,marforio2013secure,sandhu2005peer} increases security but it is expensive, is not always available or not accessible to developers~\cite{McGillionDNA15, TrustyTEE}, is harmful to usability~\cite{SDasFC18}, or is not flexible enough to run advanced protocols.

Threshold cryptography~\cite{Desmedt87,boyd1989digital}
strengthens cryptographic protocols by distributing the user's secret key among several parties. 
There exist a number of threshold encryption and signature schemes~\cite{desmedt1992shared,rabin1998simplified, gennaro2000robust,GennaroJKR07, Shoup00, almansa2006simplified, peeters2008practical, hazay2012efficient}, with versions that can run in users' personal devices~\cite{atwater2016shatter}.
Other works have tackled more complicated protocols, e.g., to distribute the user's secret key in attribute-based credentials~\cite{brands2000rethinking}, or to make threshold-cryptographic versions of zero-knowledge proofs~\cite{keller2012efficient}. 

Many works propose systems in which the user's secret key is shared between a user's device and a central server~\cite{mackenzie2001networked, camenischvirtual, boneh2001method, boneh2004fine, libert2003efficient, buldas2017server}. These schemes enable users to block their keys, but do not provide user privacy towards the central server as it is essential that users authenticate themselves to the server to enable the server to find the correct key. Adding anonymous communication or retrieval mechanisms does not resolve this privacy problem.
Camenisch et al.~\cite{camenischvirtual} ensure privacy to some extent in signature schemes by blinding the message being signed during the threshold protocol with the server. Yet, the server learns when and how often the user uses her signing key. Hence, users are vulnerable to timing attacks~\cite{JawaheriSBE18}. Brands' scheme~\cite{brands2000rethinking} protects against timing attacks as long as the shareholder cannot store a timed log of operations (e.g., when it is a smart card), but not when keys are shared with an online server.  

\name is designed to increase the privacy of these threshold-cryptographic solutions. We compare the privacy properties obtained when using \name with those in previous proposals in Table~\ref{tab:related-work}. We consider three privacy aspects: user's anonymity when running the threshold protocol (i.e., need to authenticate); data hiding (e.g., signed message) in the protocol from the server; and usage pattern hiding to avoid timing attacks. Generic schemes focus on the security of the key and thus provide no privacy. The special-purpose designs only protect data involved in the protocol.

\newcolumntype{C}[1]{>{\centering\let\newline\\\arraybackslash\hspace{0pt}}m{#1}}
\newcommand{\tableno}{$\times$}
\newcommand{\tableyes}{\checkmark}
\begin{table} 
  \centering
  \begin{threeparttable}
  \caption{Comparison of generic, special-purpose and Tandem-augmented threshold-cryptographic
    protocols (TCPS).}\label{tab:related-work}
  \begin{tabular}{lccc}
 & Generic      & S. Purpose      & with \\
     & e.g., \cite{GennaroJKR07,ShoupG02} & \cite{camenischvirtual,brands2000rethinking}      & \name \\
    \midrule
    Anonymous key usage     & \tableno & \tableno   & \tableyes \\
    Hide protocol data      & \tableno & \tableyes  & (*)  \\
    Hide key-usage patterns & \tableno & \tableno   & \tableyes \\
  \end{tabular}
  \begin{tablenotes}
   \item[*] Achieved by \name if the underlying TCP does.
  \end{tablenotes}
  \end{threeparttable}
\end{table}

Password-hardening services~\cite{EverspaughCSJR15,CamenischLN15} use decentralization to increase security of authentication servers against brute-force attacks. Similarly to \name, these schemes introduce a hardening server to rate limits or block requests from the main authentication server.  However, in the password scenario the hardening server is only accessed by the authentication server. Therefore, there are no privacy concerns and indeed these techniques do not provide any privacy protection. 

Single-password authentication schemes~\cite{AcarBK13, IslerK17, RossJMBM05} hide the user's password from potentially malicious authentication servers. Similar to \name, \cite{AcarBK13, IslerK17} use a central server to help the user to authenticate to the authentication server. These schemes, however, 
reconstruct the user's authentication key in the user's device, and thus cannot protect against active attackers. \section{Problem Statement}
\label{sec:system}

We consider a scenario where \textit{users} are required to perform cryptographic operations 
on insecure devices (i.e., without secure hardware) to interact with a \textit{service provider} (SP). 
To keep their keys safe, users use a \emph{central server} to run threshold-cryptographic protocols (TCPs).
The TCP protects the \emph{security} of users' keys against \emph{active adversaries} that can compromise a users device, 
and that can observe the key being used in protocols. As long as the central server remains honest, 
security is guaranteed: key usage can be blocked and rate-limited.

Users wish to retain the \emph{privacy} they had with respect to the SP before
using the central server with TCPs, even if the central server is malicious and
colludes with the SP. 
We call an execution of the protocol between the user and the central server a \emph{transaction}. 

\para{Security and privacy properties.} We formalize the desired security and privacy properties of the system.

\begin{property}[Key security]
\label{property:protect-key}
The central server should \emph{prevent unauthorized use of the user's key} even if the
user's device is compromised. The user must be able to \emph{block a compromised key} at the
server so that the attacker can no longer use it.
\end{property}
Any solution that recomputes the full user's key on the user's device, e.g., by
deriving it from a user-entered passphrase or by interacting with a central
server~\cite{AcarBK13,IslerK17}, does not satisfy this property: an attacker who compromises the user's device
can observe the full key. Thereafter, the attacker can use the key indefinitely, making blocking impossible.

\begin{property}[Key-use privacy]
\label{property:privacy}
The user must have \emph{privacy of key use in transactions}.
The central server must not be able to distinguish between two users performing transactions even 
if it colludes with the service provider (SP). 
(Unless the SP could distinguish the users, in which case collusion leads to a trivial and unavoidable privacy breach).
\end{property}

\begin{property}[Key rate-limiting]
 \label{property:rate-limit}
 Users must be able to \emph{limit the rate of usage of their keys} in a given interval of time called an \emph{epoch}.
\end{property}

\para{Threat model.}
We assume the central server is available, and follows the protocols to enforce 
key security and rate-limits on the users' keys (\cref{property:protect-key} and
\cref{property:rate-limit}) in the face of active attackers.
Servers that violate this assumption cannot use the user's key as long as the user's device
remains honest. 

For privacy, the central server may be malicious, i.e., interested in breaching
the privacy of users by trying to learn which keys and services they use
(\cref{property:privacy}). It can collude with the SP.

\para{Why na\"ive centralization does not work.} 
Consider a user that secret-shares her key with the central server. When she needs to run a threshold-cryptographic protocol, she authenticates against the central server
and jointly executes the TCP with it.
This scheme offers key security (\cref{property:protect-key}): The server alone cannot use the user's key and if an attacker compromises the user's device the user can authenticate to the central server and request blocking. It also provides key rate-limiting (\cref{property:rate-limit}): the central server can easily enforce a limit on the number of times the key is used. However, since the user is identified while using the key, the scheme does not achieve key-use privacy (\cref{property:privacy}).

The lack of key-use privacy has implications when the interactions between the user and the SP are anonymous (e.g., showing an anonymous credential). An SP colluding with the central server can exploit time correlations between the times when the authenticated user interacts with the central server and when the anonymous user interacts with the SP to de-anonymize the user. The anonymity set of the user is reduced to the authenticated users interacting with the central server around the transaction time. This attack has been used in the early days of Tor to identify users and hidden services~\cite{AbbottLLP07,OverlierS06}. As this attack relies solely on time correlation between accesses, it cannot be prevented by making the messages seen by the central server and the SP during the TCP cryptographically unlinkable~\cite{brands2000rethinking}.

Straightforward approaches to prevent time-correlation attacks such as introducing delays and introducing dummy requests are difficult to use in practice. Either operations need to be delayed for a long time, ruling out real-time applications such as showing an anonymous credential or performing a payment; or impose high overhead on the users and the central server, and are difficult to generate \cite{BalsaTD12,ChowG09}.

\section{\name at a Glance}

We introduce \name, a set of protocols that wrap threshold-cryptographic schemes to provide privacy when the user's key is shared with a central server (\emph{the \name server, \TS}). 
\name ensures key security and key rate-limiting, formalized in Game~\ref{game:security}, and key-use privacy, formalized in Game~\ref{game:privacy}, if the \TS colludes with service providers (SPs), and Definition~\ref{def:honest-sp-privacy}, if it does not. 

For simplicity, we assume that there is only one \name server (\TS). Secret-sharing the key with multiple \name servers would increase security and/or robustness, without detriment to privacy. We explain how to do so in Appendix~\ref{sec:appendix-multiple-ts}. We also assume that users can use an anonymous communication channel~\cite{PiotrowskaHEMD17, TorDingledineMS04} to communicate with the \TS and SPs to protect their privacy at the network layer.

 \begin{figure}
 \includegraphics[trim={3mm 1mm 0 0},clip,scale=0.98]{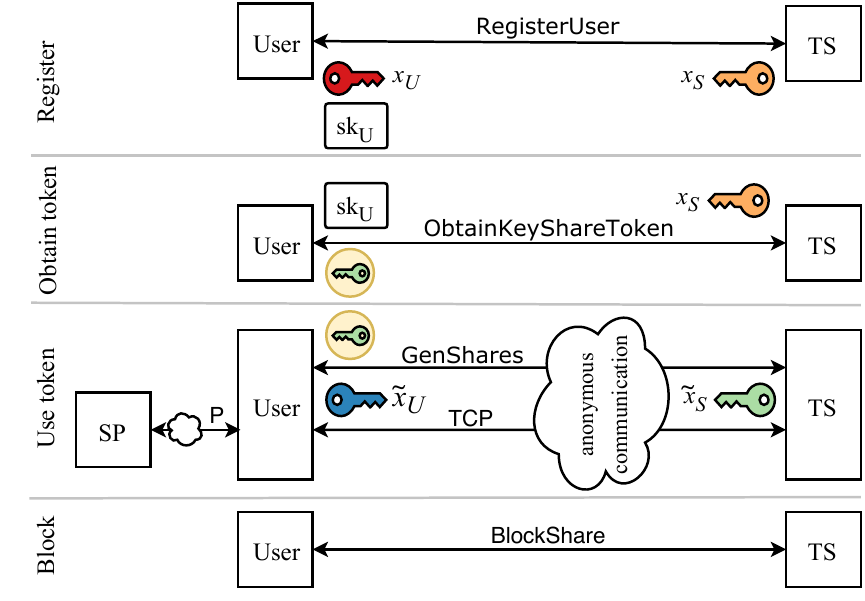}

 \newcommand{\keysharetokenimg}{\raisebox{-0.8mm}{\includegraphics[trim={2mm 0.6mm 0 0},clip,scale=0.7]{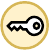}}}
 \caption[nothing]{\label{fig:tandem-overview} During registration, the user and \TS derive key-shares $\xp$ and $\xs$ of a secret $\secret = \xp + \xs$. The user stores its authentication key $\abcsku$. Users authenticate before obtaining key-share tokens (\keysharetokenimg), containing a randomized server key share $\xsfresh$. To use her key $\secret$ anonymously, the user connects to the \TS via an anonymous channel and sends a key-share token. The \TS recovers the key share $\xsfresh$ and uses it to run the \TCP with the user share $\xpfresh$ (the user and SP run protocol \tcprp). The user can block her key-share tokens at any time without any key. Inputs are shown above the arrows, outputs below.} 
\end{figure}

Suppose that Alice uses Schnorr's identification protocol to authenticate
herself to her bank. Let $\grouporder$ be the group order. 
To prevent others from accessing her account, Alice protects her secret key $\secret$ using the threshold-cryptographic version of Schnorr's protocol in \cref{fig:threshold-schnorr}. The bank knows Alice's public key $\schnorrpk$. Alice creates shares $\xs$ and $\xp$ such that $\secret = \xs + \xp$. She gives $\xs$ to a \name server, keeping $\xp$ on her device. 
We sketch how \name can wrap the threshold-cryptographic Schnorr protocol so that Alice retains the advantage of using a central server, without revealing her actions to this party.
Fig.~\ref{fig:tandem-overview} illustrates the process.

\name uses an additive homormorphic encryption scheme $(\penc, \pdec)$~\cite{JoyeL13}, and a blind signature scheme~\cite{BaldimtsiL13,BBS+AuSM06}. The \TS generates key-pairs $(\ppk, \psk)$ and $(\bspk, \bssk)$ for these respective schemes and publishes $\ppk$ and $\bspk.$

\begin{figure}[tb]
  \centering
    \begin{tabular}{lclcl}
      Server & & User & & SP \\
      {\unboldmath$\xsfresh \in \Zp$} & & {\unboldmath$\xpfresh \in \Zp$} & & {\unboldmath$\schnorrpk = \generator^{\secret}$} \\
      \midrule
 
      $\serrand \randin \Zp$ & & $\phrand \randin \Zp$ & &\\ 
      $\sercommit=\generator^{\serrand}$ & $\diagramshortsend{\sercommit}$ & $\phcommit = \generator^{\phrand} $ \\
      & & $u = \sercommit \phcommit$ & $\diagramshortsend{u}$ \\
      & $\diagramshortrecv{\chal}$ & &$\diagramshortrecv{\chal}$ & $\chal \in_R \Zp$\\
      $\serResp=\serrand+\chal\xsfresh$ & $\diagramshortsend{\serResp}$ & $\phresp = \phrand +\chal\xpfresh$ & & \\
      & & $\xusresp = \phResp+\serResp$ & $\diagramshortsend{\xusresp}$ & $u \? \generator^{\xusresp}h^{-\chal}$\\
      \end{tabular}
      \caption{\label{fig:threshold-schnorr} A threshold version of Schnorr's proof of identity. The user and the server respectively hold the (fresh) shares $\xpfresh$ and $\xsfresh$ of the private key $\secret = \xpfresh + \xsfresh$ corresponding to the public key $\schnorrpk = \generator^{\secret}$. They jointly compute a proof of knowledge of $\secret$ such that $\schnorrpk = \generator^{\secret}$. We write \textsf{TCP} for the threshold-cryptographic protocol between the user and the server, and \textsf{P} for the protocol between the user and the SP.}
\end{figure}

\para{Registration.}
Alice registers with the \TS using the \registeruser protocol. 
During registration, Alice and the \TS jointly compute long-term shares $\xp$ and $\xs$ of a long-term key $\secret$. Alice generates a public-private key pair $(\abcpku, \abcsku)$, and sends $\abcpku$ to the \TS. 
The \TS sends an homomorphic encryption $\xsenc = \penc(\xs)$ to Alice.
Because $\xsenc$ is encrypted against the \TS's key, Alice does not learn anything about the share $\sks{}$.

\para{Obtain token.}
Key-share tokens enable Alice to later anonymously use her key (see below). To obtain a token, Alice runs the \obtainkstoken protocol with the \TS. First, she uses $\abcsku$ to authenticate to the \TS. Then, Alice and the \TS construct a one-time-use key-share token containing a randomized version of the \TS's key share $\xs$.
To do so, Alice picks a large $\xsencdelta$ and computes $\ctxt = \xsenc \cdot \penc(\xsencdelta) = \penc(\xs + \xsencdelta)$. Then, she sends to the \TS a commitment $\bscommitment$ commiting to both $\ctxt$ and her key $\abcsku$. She proves that the committed ciphertext $\ctxt$ was constructed by additively randomizing $\xsenc$ and that $\bscommitment$ contains the correct private key. If the proof is correct, the \TS blindly signs the commitment and Alice receives a signature $\bssignature$ on a rerandomized commitment $\bsfreshcommitment$. Alice stores these in her one-time-use key-share token $\ksstoken = (\bssignature, \bsfreshcommitment, \ctxt, \xsencdelta)$.
The \TS can limit the number of issued tokens, enforcing a rate limit on Alice's key.

Key-share tokens may seem similar to passwords: both unlock functionality. However, unlike passwords, key-share tokens can be verified and used \emph{without} knowing the user's identity. Key-share tokens contain a randomized key share $\xsfresh = \pdec(\ctxt)$ essential for the TCP. Hence, \name cannot be replaced by a password-hardening service~\cite{EverspaughCSJR15,CamenischLN15}.
The randomized key shares contained in the tokens also distinguish them from traditional eCash tokens~\cite{ChaumFN88,CamenischHL05,MiersG0R13}.

\para{Using keys.}
After obtaining a token $\ksstoken = (\bssignature, \bsfreshcommitment, \ctxt, \xsencdelta)$, Alice can anonymously run threshold-cryptographic protocols with the \TS using the \genshares protocol. Alice contacts the \TS via an anonymous channel, and sends the signature $\bssignature$, the commitment $\bsfreshcommitment$, and the randomized encryption $\ctxt$ of the \TS' key share. She proves that $\ctxt$ is committed to in $\bsfreshcommitment$ and that the private key in $\bsfreshcommitment$ does not correspond to a blocked public key. The \TS cannot recognize the blindly signed commitment nor the ciphertext $\ctxt$ it contains.

If the signature and proof are correct, the \TS derives the fresh share $\xsfresh = \pdec(\ctxt) \pmod{\grouporder} = \xs + \xsencdelta \pmod{\grouporder}$. The size of $\xsencdelta$ ensures that $\xsfresh$ cannot be linked to the long-term share $\xs$, thus hiding Alice's identity from the \TS. Alice derives her fresh share $\xpfresh = \xp - \xsencdelta \pmod{\grouporder}$. By construction, $\secret = \xsfresh + \xpfresh = \xp + \xs \pmod{\grouporder}$.

Alice can now run the \textsf{TCP} protocol in Fig.~\ref{fig:threshold-schnorr} to prove her identity to the bank. The \TS and Alice use the just computed respective fresh shares $\xsfresh$ and $\xpfresh$. The \TS does \emph{not} learn which user ran the \textsf{TCP} protocol.
Note that the \TS never communicates directly with service providers. Therefore, Alice can use \name without the SP's knowledge.

\para{Blocking keys.}
Alice can request the \TS to block her key by using the \blockshare protocol. She authenticates to the \TS using any pre-defined means (e.g., a PUK, or a passphrase; knowledge of the private key $\abcsku$ is \emph{not} required). The \TS adds Alice's public key $\abcpku$ to the list of blocked keys, causing existing tokens to become invalid, and she cannot create new tokens.

Storing the private key $\abcsku$ and unused tokens on the user's device is much safer than storing the full key directly. Even if an attacker obtains $\abcsku$ and unused tokens, \name guarantees key security and key rate-limiting. Running \blockshare immediately invalidates existing tokens, and prevents the attacker from obtaining new ones. As a result, the attacker can no longer use the user's key $\secret$. This is not the case if the attacker can obtain $\secret$ directly. 

\para{Preventing time correlation.}
To preserve her privacy, the actions of obtaining tokens---where Alice is authenticated---and using tokens---where Alice is anonymous---\emph{must} be uncorrelated, i.e., tokens should not be obtained right before being used. To avoid correlation, Alice can configure her device to obtain tokens at random or regular times (e.g., every night), ensuring that tokens are always available. Suppose Alice obtains fresh tokens every morning, then a time-line of
registration ($\treg$), obtaining tokens ($\tobtain{i}$), using tokens ($\tspend{i})$, and blocking the key ($\tblock$) events might look as follows:\\
\begin{center}
\vspace{-8mm}
\begin{tikzpicture}
\coordinate (T0) at (0,0);
\coordinate (T1) at (0.5,0);
\coordinate (T2) at (3.75,0);
\coordinate (T3) at (7,0);
\coordinate (T4) at (8,0);

\draw[->] (T0) to node[pos=0.95,above]{time} (T4);
\draw[->] (T0) to node[pos=0.95,above]{time} (T4);

\newcommand{\event}[3]{
  \coordinate (#1) at ($ #2 $);
  \draw ($ (#1) + (0,2pt) $) to ($ (#1) + (0, -2pt) $);
  \node (node#1) at ($ (#1) + (0, -8pt) $) {#3};
}
\event{t2}{(T1)!.0!(T3)}{$\treg$}
\event{t3}{(T1)!.3!(T3)}{$\tobtain{1}$}
\event{t4}{(T1)!.35!(T3)}{$\tobtain{2}$}
\event{t5}{(T1)!.4!(T3)}{$\tobtain{3}$}

\event{s3}{(T1)!.70!(T3)}{$\tspend{1}$}
\event{s2}{(T1)!.77!(T3)}{$\tspend{2}$}

\event{b}{(T3)!0.25!(T4)}{$\tblock$}

\end{tikzpicture} \vspace{-4mm}	
\end{center}
Note that the obtain and use events do not necessarily follow each other and are \emph{not} correlated. The third token $\tobtain{3}$ is unused when Alice blocks her key at time $\tblock$. This token can thereafter not be used.

\section{Cryptographic Preliminaries}
\label{sec:preliminaries}

Let $\secpar$ be a security parameter. Throughout, $\G$ is a cyclic group of prime order $\grouporder$ (of $2\secpar$ bits) generated by $\generator.$ We write $\Zp$ for the integers modulo $\grouporder$; by $\NatNumUpTo{n}$ we denote the set $\{0, \ldots, n - 1\}$; and by $a \randin A$ we denote that $a$ is chosen uniformly at random from the set $A$. We use a cryptographic hash function $\hash: \{0,1\}^* \to \Zp$ that maps strings to integers modulo $\grouporder$.
For reference, Table~\ref{tab:building-blocks} in Section~\ref{keyshare-tokens} summarizes the notation used by \name's building blocks, and Table~\ref{tab:notation-tandem} in Section~\ref{keyshare-tokens} explains frequently-used symbols in \name. 

\subsection{Cryptographic Building Blocks}
\name relies on a couple of cryptographic building blocks. We use an additive homomorphic encryption scheme given by the algorithms $\pkeygen$, $\penc, \pdec$ with plaintext space $\mathbb{Z}_{\PaillierN}$
and space of randomizers $\randspace$. Let $(\ppk, \psk) = \pkeygen(1^{\secpar})$ be a key-pair, then we write $\ctxt = \penc(m; \prand)$ to denote the homomorphic encryption of the message $m \in \mathbb{Z}_{\PaillierN}$ using randomness $\prand \in \prandspace$. The scheme is additively homomorphic, so
\begin{equation*}
  \mathbf{E}^+_{pk}(m_1 ; \prand_1) \mathbf{E}^+_{pk}(m_2; \prand_2) =
  \mathbf{E}^+_{pk}(m_1 + m_2\;(\textrm{mod}\,{\PaillierN}) ; \prand_1 \prand_2).
\end{equation*}
Our proof of concept uses Joye and Libert's encryption scheme~\cite{JoyeL13}, see Appendix~\ref{sec:appendix:joye-libert}, but Paillier's scheme~\cite{Paillier99} would also work.

\name uses two computationally hiding and binding commitment schemes. First, by $\commit(m,r)$ we denote a commitment function that takes a message $m \in \Zp$ and a randomizer $r \in \Zp$. Analogously, we define $\commit((m_1,\allowbreak \ldots,\allowbreak m_k),\allowbreak r)$ to commit to a tuple of messages. We instantiate this scheme using Pedersen's commitments~\cite{Pedersen91}. Second, we denote by $\Delta = \extcommit{m}{r}$ with $m \in \{0,1\}^*, r \in \{0,1\}^{2\secpar}$ an extractable commitment scheme~\cite{SantisCP00}. That is, in our reductions, we can extract the input $m$ used to create a commitment $\Delta$. For example, the instantiation $\extcommit{m}{r} = H(m \| r)$ is extractable in the random oracle model for $H$.

To construct key-share tokens the \TS signs them using a blind signature scheme supporting attributes~\cite{BaldimtsiL13} given by the following protocols:
\begin{itemize}
\item The signer runs $(\bspk, \bssk) = \bssetup(1^\secpar, \nrattributes)$ to setup a system for signatures on $\nrattributes$ attributes. It obtains a public-private key pair $(\bspk, \bssk)$.
\item The interactive protocol $\bssign(\bspk, C)$ is run by a user and the signer on input of the signer's public key $\bspk$ and a Pedersen commitment $\bscommitment = \commit((a_1,\allowbreak \ldots,\allowbreak a_{\nrattributes}), \bscommitmentrand)$ to the attributes. The signer takes its private key $\bssk$ as private input, whereas the user takes the attributes $(a_1, \ldots, a_{\nrattributes})$ and the randomizer $r$ as private input. At the end of protocol, the user obtains the tuple $(\bssignature, \bsfreshcommitment, \bsfreshcommitmentrand)$ where $\bssignature$ is a signature on $\bsfreshcommitment = \commit((a_1, \ldots, a_{\nrattributes}), \bsfreshcommitmentrand)$, a fresh commitment to the attributes. The issuer does not learn the values of the attributes nor the resulting signature $\tokensignature$.
\item The verifier calls $\bsverify(\bspk, \bssignature, \bsfreshcommitment)$ to verify the signature $\bssignature$ on commitment $\bsfreshcommitment$. The algorithm outputs $\top$ if the signature is valid, and $\bot$ otherwise.
\end{itemize}
We require that the scheme has the blind signing property~\cite{BaldimtsiL13}; and that signatures are unforgeable~\cite{BaldimtsiL13}. For example, the scheme by Baldimtsi and Lysyanskaya~\cite{BaldimtsiL13} satisfies these properties.

\subsection{Threshold-Cryptographic Protocols}
\label{sec:tcps}
In this paper, we focus on cryptographic protocols run between a user and a service provider, e.g., showing a credential to an SP or spending an electronic coin. 
The threshold-cryptographic version of such a protocol splits the user's key $\secret$ and the user's side of the original protocol in two parts, run by different parties. Each party operates on a secret-share of the user's key. Security of the threshold-cryptographic protocol (TCP) ensures that a large enough subset of shares (two in the case of two parties) are required to complete the protocol.

We focus on TCPs where the user's side of the protocol is distributed between the user and the \TS. After registration, the user and the \TS hold the shares $\xp$ and $\xs$ of $\secret$. After running $\genshares$ with a new token, the user and \TS hold fresh key shares $\xpfresh$ and $\xsfresh$. They then run the TCP protocol, which we denote as:
\begin{equation}
  \label{eq:tcp-protocol}
  \tcprp(\auxrp) \leftrightarrow
  \tcpu(\xpfresh, \auxu) \leftrightarrow
  \tcpts(\xsfresh),
\end{equation}
where the SP, the user and the \TS respectively run the interactive programs $\tcprp$, $\tcpu$ and $\tcpts$. The user mediates all interactions between the service provider and the \TS.
The user and the SP take extra inputs needed for the execution of the target cryptographic protocol denoted as $\auxu$ and $\auxrp$.
We denote the complete protocol from~\eqref{eq:tcp-protocol} by $\tcp(\xpfresh, \xsfresh, \auxu, \auxrp)$.

\name can only enhance the privacy (\cref{property:privacy}) of certain TCPs. We formalize the condition that these TCPs should satisfy. To avoid that the \TS can recognize the user based on the shares input to the TCP, we randomize long-term secret shares. Thus, we require that TCPs enhanced with \name function with randomized key shares. In addition, our privacy-friendly \genshares protocol requires this randomization to be linear.

For simplicity, we assume that the user's secret $\secret \in \Zp$ for some field $\Zp$ of prime order $\grouporder$ (e.g., corresponding to the group $\G$ defined above). Our constructions, however, can be modified to settings with unknown order arising from RSA assumptions. Formally, we require the TCP to be linearly randomizable:

\begin{definition}
  Let $\xp, \xs \in \Zp$ be secret shares of the user's secret $\secret$. 
  Then, we say that the TCP is \emph{linearly randomizable} if
  for all $\kssrandomizer$ we have that
  (1) if $\tcp(\xp, \xs, \auxu, \auxrp)$ completes successfully, then so does $\tcp(\xp - \kssrandomizer, \xs + \kssrandomizer, \auxu, \auxrp)$, and (2) $\xs + \kssrandomizer$ is independent from $\xs$.
\end{definition}

\para{}The first condition implies that the original secret sharing $(\xp, \xs)$ and the randomized secret sharing $(\xp - \xsencdelta, \xs + \xsencdelta)$ must share the same secret, whereas the second implies that the \TS cannot recognize the user from the randomized secret share alone.

\para{Security and privacy properties of TCPs.} To ensure that a TCP with \name satisfies the security properties (\cref{property:protect-key} and~\cref{property:rate-limit}) we require that the TCP itself is secure. That is, if the \TS no longer uses its share $\xs$ to run its part of the TCP, 
then no malicious user can successfully complete the TCP with the SP. We formalize this in Game~\ref{game:tcp-security} in Section~\ref{sec:security-privacy}.

To ensure that a TCP with \name satisfies the privacy property (\cref{property:privacy}) we require that the TCP itself is private with respect to the \TS (respectively the \TS colluding with the SP): If the \TS runs its part of the TCP  
using a randomized key-share as input, then the \TS (respectively the \TS and the SP) cannot recognize the user. We formalize this in Game~\ref{game:tcp-privacy} in Section~\ref{sec:security-privacy}.
 \section{The Full \name Construction}
\label{keyshare-tokens}

We now introduce the full \name construction.

\begin{table}[tb]
  \caption{\label{tab:building-blocks}Notation and cryptographic building blocks.}
\begin{tabularx}{\columnwidth}{lX@{}}  
\textbf{Symbol}    & \textbf{Interpretation}  \\
\midrule
$\NatNumUpTo{n}$ & The set $\{0,\ldots,n - 1\}$ \\
$\secpar$ & The security parameter \\
$\G, \generator, \grouporder$ & Cyclic group $\G = \langle \generator \rangle$ of order $\grouporder$\\[2mm]
\multicolumn{2}{@{}l}{\emph{Additively homomorphic encryption scheme}} \\
$\pkeygen(1^{\secpar})$ & Generate public-private key-pair \\
$\penc(m; r)$   & Encrypt $m \in \mathbb{Z}_{\PaillierN}$ with randomizer $r \in \prandspace$\\ 
$\pdec(\ctxt)$  & Decrypt ciphertext $\ctxt$ \\
  $\PaillierN$ & Size of additive plaintext domain \\
  $\prandspace$ & Space of randomizers \\[2mm]
\multicolumn{2}{@{}l}{\emph{Commitment schemes and hash function}} \\
$\commit(m, r)$ & Commit to $m \in \Zp$ (or a tuple of messages) with randomizer $r \in \Zp$ \\
$\extcommit{m}{r}$ & Commit to $m \in \{0,1\}^*$ with randomizer $r \in \{0,1\}^{2\ell}$ \\
$\hash(s)$ & Hash function from $s \in \{0,1\}^*$ to $\Zp$ \\[2mm]
\multicolumn{2}{@{}l}{\emph{Blind signature scheme with attributes}} \\
$\bssetup(1^{\secpar}, \lambda)$ & Generate signer's key-pair for signatures on $\lambda$ attributes \\
$\bssign(\bspk, \bscommitment)$ & Protocol to blindly sign attributes in $C$ \\
$\bsverify(\bspk, \bssignature, \bsfreshcommitment)$ & Verify signature $\bssignature$ on $\bsfreshcommitment$.\\
\end{tabularx} 
\end{table}

\begin{table}[tb]
  \caption{\label{tab:notation-tandem} Frequently used symbols in \name protocols}
  \centering
\begin{tabular}{ll}  
\textbf{Symbol}    & \textbf{Interpretation}  \\
\midrule
$\tokendisclose$ & Disclose subset in cut-and-choose construction \\
$\xsencdelta,\kssrandomizer_i$ & Randomizers of key shares \\
$\tokensecpar$ & Token security parameter \\
$\lengthdelta$ & Length of randomizers $\kssrandomizer_i$ in bits \\
$\secret$      & Long-term secret key for a user    \\
$\ppk, \psk$ & Public-private encryption key-pair of \TS\\
$\bspk, \bssk$ & Public-private signing key-pair of \TS\\
$\abcpku, \abcsku$  & Public-private key-pair of the user $\User$  \\
$\grouporder$ & Order of the group $\G$ \\
$\xp, \xs$ &  Long-term key share held by user resp. \TS \\
$\xsenc$ & Homomorphic encryption of $\xs$ \\
$\xpfresh, \xsfresh$ & User's resp. TS' key share output by \genshares \\
$\epoch$ & The current epoch\\
$\tokensignature$ & Blind signature of the \TS \\
\end{tabular}
\end{table}

\para{Setup.}
The \TS sets up the \name system as follows.
\begin{protocol}
The $\tandemsetup(1^\secpar,1^{\tokensecpar})$ protocol is run by the \TS, where $\secpar$ and $\tokensecpar$ are security parameters. The \TS generates a public-private key-pairs $(\ppk, \psk) = \pkeygen(1^{\secpar})$ for the homomorphic encryption scheme and $(\bspk, \bssk) = \bssetup(1^{\secpar}, \tokensecpar + 3)$ for the blind signature scheme. The \TS publishes $\ppk$ and $\bspk$. Finally, the \TS keeps track of an initially empty list of revoked public keys $\revocationlist$.
\end{protocol}

\para{Registering users.}
When a user first registers at the \TS, the \TS computes a key-share $\xs$ for that user, and sends her an encrypted version $\xsenc = \penc(\xs)$. To ensure that the \TS cannot hide an identifier in higher-order bits of $\xs$ that are not randomized by the user in the remainder of the protocol, the \TS proves that the plaintext $\xs$ is in the correct range.

\begin{protocol}
The \registeruser protocol is run between a user and the \TS, and proceeds as follows.
\begin{enumerate}
\item The user opens an encrypted channel to the \TS and authenticates it.
\item
  \label{step:register:choose-secrets}
  The user $\User$ and the \TS generate secret shares $\xp \randin \Zp$ and $\xs \randin \Zp$, respectively.
  The user also generates a public-private key-pair $(\abcpku, \abcsku) = (\generator^{\abcsku}, \abcsku)$ for $\abcsku \randin \Zp$ that we use to authenticate the user's device and to revoke the user's tokens. The user sends $\abcpku$ to the \TS.
\item The \TS picks $\prand \randin \prandspace$, computes $\xsenc = \penc(\sks{}; \prand)$
  and sends $\xsenc$ to the user. Moreover, the \TS sends a range proof to the user that $\xsenc$
  is constructed correctly, i.e., that
    \begin{equation}
      \label{eq:range-proof-xs}
      \pdec(\xsenc{}) \in [0, \grouporder).
    \end{equation}
    \fullversioncmd{See Appendix~\ref{sec:appendix-register-proof} for how to instantiate this proof.}
\begin{confversion}
      For Paillier's encryption scheme, this proof can be instantiated using standard techniques~\cite{BellareG97}. The Appendix of the full version~\cite{fullversion} constructs this proof for Joye and Libert's scheme~\cite{JoyeL13} using a cut-and-choose technique.
\end{confversion}
 \item The \TS records $(\sks{}, \xsenc, \abcpku)$ for this user, and marks
   this user as active.
   The user stores $(\sku{}, \xsenc, \abcsku)$ on her device.
\end{enumerate}
\end{protocol}
In Appendix~\ref{sec:appendix-multiple-ts} we explain how users can use multiple \TSs to increase robustness and/or security.

\para{Obtain a key-share token.}
First, the user randomizes the ciphertext $\xsenc$. However, it seems difficult to prove directly, for example in zero-knowledge, that the randomized ciphertext produced by the user is of the correct form. 
Therefore, we use a standard cut-and-choose approach~\cite{BrassardCC88,ChaumFN88} to allow the \TS to check that the key share it uses in the \TCP is a randomization of the correct secret key with overwhelming probability.

Let $\lengthdelta$ be a security parameter.
The user constructs
$2\tokensecpar$ witness ciphertexts $\ctxt_i = \xsenc \cdot \penc(\kssrandomizer_i; \prand_i)$ with $\kssrandomizer_i \randin [0, 2^{\lengthdelta})$ and $\prand_i \randin \prandspace$. The users sends commitments $\ksscommitment_i$ to these ciphertexts to the \TS. The \TS then asks the user to open a subset $\tokendisclose$ of cardinality $\tokensecpar$, so that the \TS can verify that these $\tokensecpar$ ciphertexts were correctly formed. The user then picks $\xsencdelta \randin [2^{\lengthdelta}, 2^{\lengthdelta + 1})$ and $\prand \randin \prandspace$ and constructs the randomized ciphertext $\ctxt = \xsenc \cdot \penc(\xsencdelta; \prand)$. The choice of $\xsencdelta$ ensures that it is always bigger than the $\kssrandomizer_i$s. The user constructs a commitment to her private key $\abcsku$, the current epoch $\epoch$,
$\ctxt$, and the remaining $\tokensecpar$ unopened witness ciphertexts $\ctxt_i$. The \TS blindly signs this commitment.

\label{ref:size-lengthdelta}
We set $\lengthdelta = \lengthdeltaval$ to ensure that the $\tokensecpar + 1$
plaintexts $\xs + \kssrandomizer_i$ and $\xs + \xsencdelta$ corresponding to the
ciphertexts in the token statistically hide $\xs$. We require that
the size of plaintext space $\PaillierN$ of $\penc$ is
bigger than $2^{\lengthdelta + 2}$ to ensure no overflows occur.

\begin{protocol}
  The \obtainkstoken protocol is run between a user and the \TS.
\begin{enumerate}
\item The user opens an encrypted channel to the \TS and authenticates it.
\item The user recovers $(\sku{}, \xsenc, \abcsku)$
  from storage, and authenticates to the \TS using $\abcsku$.
  The \TS looks up the corresponding user's record $(\sks{}, \xsenc, \abcpku)$
  and aborts if this user exceeded the rate-limit for the current epoch, was banned, or was blocked. 
\item The \TS picks a random subset $\tokendisclose \subset \{1, \ldots, 2\tokensecpar\}$ of cardinality $\tokensecpar$ of indices of ciphertexts it will check at step 5; and commits to $\tokendisclose$ by picking $\discloseSetCommitmentRand \randin \{0,1\}^{2\secpar}$ and sending $\discloseSetCommitment = \extcommit{\tokendisclose} {\discloseSetCommitmentRand}$ to the user.
\item The user picks randomizers $\kssrandomizer_1, \ldots, \kssrandomizer_{2\tokensecpar} \in \{0,1\}^{\lengthdelta}$ and $\prand_1, \ldots, \prand_{2\tokensecpar} \in \prandspace$ to create witness ciphertexts; and randomizers $\ksscommitmentrand_1, \ldots, \ksscommitmentrand_{2\tokensecpar} \in \Zp$ and $\deltacommitmentrand_1, \ldots, \deltacommitmentrand_{2\tokensecpar} \in \{0,1\}^{2\secpar}$ for the commitments and sets:
  \begin{equation}
    \begin{split}
    \ctxt_i      &= \xsenc \cdot \penc(\kssrandomizer_i; \prand_i) \\
    \ksscommitment_i &= \commit(\hash(\ctxt_i), \ksscommitmentrand_i) \\
    \deltacommitment_i &= \extcommit{(\kssrandomizer_i, \prand_i)}{ \deltacommitmentrand_i},
    \end{split}
    \label{eq:token-element}
  \end{equation}
  for $i = 1, \ldots, 2\tokensecpar$. She sends the commitments $\ksscommitment_1, \ldots, \ksscommitment_{2\tokensecpar}$ and $\deltacommitment_1, \ldots, \deltacommitment_{2\tokensecpar}$ to the \TS. Note that the commitments $\ksscommitment_i$ and $\deltacommitment_i$ are computationally binding and hiding. We use the extractable commitments $\deltacommitment_i$ to extract the inputs $\kssrandomizer_i, \prand_i$ in the proofs.
\item The \TS opens the commitment $\discloseSetCommitment$ by sending the subset $\tokendisclose$ and the randomizer $\discloseSetCommitmentRand$ to the user. The user checks that $\discloseSetCommitment = \extcommit{\tokendisclose} {\discloseSetCommitmentRand}$, and aborts if the check fails.
  \item
The user sends 
    $(\ctxt_i, \kssrandomizer_i,\allowbreak \prand_i, \ksscommitmentrand_i, \deltacommitmentrand_i)_{i \in \tokendisclose}$ to the \TS to open the requested commitments.
The \TS checks that these values satisfy equation~\eqref{eq:token-element} and that $\kssrandomizer_i < 2^{\lengthdelta}$. If any check fails, the \TS bans the user.
\item
  \label{step:obtaintoken:pick-delta}
  Next, the user creates the randomized ciphertext $\ctxt = \xsenc \cdot \penc(\xsencdelta; \prand)$ for $\xsencdelta \randin [2^{\lengthdelta}, 2^{\lengthdelta + 1})$ and $\prand \randin \prandspace$. Let $\tokenhidden = \{i_1, \ldots, i_{\tokensecpar}\} = \{1,\allowbreak \ldots,\allowbreak 2\tokensecpar\} \setminus \tokendisclose$ be the set of indices of unopened commitments. The user picks $\ksscommitmentrand \randin \Zp$ and sends to the \TS the commitment
    \begin{equation*}
      \ksscommitment = \commit((\abcsku, \epoch, \hash(\ctxt), \hash(\ctxt_{i_1}), \ldots, \hash(\ctxt_{i_{\tokensecpar}})), \ksscommitmentrand)
    \end{equation*}
    to her private key $\abcsku$, the epoch $\epoch$, the ciphertext $\ctxt$, and the unopened witness ciphertexts.
    Finally,
    letting $\eta = \hash(\ctxt)$ and $\eta_i = \hash(\ctxt_{i})$,
she proves in zero-knowledge to the \TS that commitment $\ksscommitment$ is correct:
    \begin{multline*}
      \textsf{PK}\{ ( (\eta_i, \ksscommitmentrand_i)_{i \in \tokenhidden},
      \abcsku, \eta, \ksscommitmentrand) \;:\; \\
      \forall i \in \tokenhidden \left[ \ksscommitment_{i} = \commit( \eta_i , \ksscommitmentrand_{i})\right] \land
      \abcpku = \generator^{\abcsku} \land \\
      \ksscommitment = \commit((\abcsku, \epoch, \eta, \eta_{i_1},\ldots,\eta_{i_\tokensecpar}), \ksscommitmentrand) \}.
    \end{multline*}
    The \TS checks this proof.
  \item If any check fails, the \TS bans the user and aborts the protocol. If all checks pass, the \TS runs $\bssign(\bspk, \ksscommitment)$ with the user. The $\TS$ takes as private input its signing key $\bssk$, the user takes as private input the attributes, and $\ksscommitmentrand$. Finally, the user obtains the tuple $(\bssignature, \kssfreshcommitment, \kssfreshcommitmentrand)$ where $\bssignature$ is a blind signature on the commitment $\kssfreshcommitment = \commit((\abcsku, \epoch, \hash(\ctxt), \hash(\ctxt_{i_1}), \ldots, \hash(\ctxt_{i_{\tokensecpar}}), \kssfreshcommitmentrand)$. The user stores the key-share token $\ksstoken = \ksstokencontentWithRange{i \in \tokenhidden}$.
\end{enumerate}
\end{protocol}

The following lemma states that even if a user is malicious, at least one of the witness ciphertexts $\ctxt_i$ must be correctly formed.
\confversioncmd{(See the Appendix of the full version~\cite{fullversion} for the proof.)}
\fullversioncmd{(See Appendix~\ref{app:proofs-of-lemmas} for the proof.)}
\begin{lemma}
  \label{lem:one-good-element}
  Consider a token $\ksstoken = \ksstokencontent$ obtained using the above protocol by a (potentially malicious) user with corresponding encrypted $\TS$ key-share $\xsenc$. Let $\deltacommitment_1, \ldots, \deltacommitment_{\tokensecpar}$ be the set of corresponding commitments used during the obtain step. Then, with probability $1 - 1/\binom{2\tokensecpar}{\tokensecpar}$ there exists an index $i^*$; and randomizers $\kssrandomizer^* < 2^{\lengthdelta}$, $\prand^*$, and $\deltacommitmentrand^*$ such that:
  \begin{align*}
    \ctxt_{i^*} &= \xsenc \cdot \penc(\kssrandomizer^*; \prand^*) \\
    \deltacommitment_{i^*} &= \extcommit{(\kssrandomizer^*, \prand^*)}{\deltacommitmentrand^*}.
  \end{align*}
\end{lemma}

Using a homomorphic-CCA secure~\cite{PrabhakaranR08} scheme with targeted malleability that allows adding known randomizers only would obviate the need for extractable commitments. Unfortunately, to the best of our knowledge no such schemes exists. The RCCA scheme by Canetti et al.~\cite{CanettiKN03} is not homomorphic, the schemes by Prabhakaran and Rosulek~\cite{PrabhakaranR08} are multiplicatively homomorphic, and the fully homomorphic scheme by Lai et al.~\cite{LaiDMSW16} is not homomorphic-CCA.

\para{Using a key-share token.}
When using a token $\ksstoken = \ksstokencontent$, the user sends
  $\bssignature,\kssfreshcommitment,\epoch,\ctxt, \ctxt_{1},\allowbreak \ldots,\allowbreak \ctxt_{\tokensecpar}$
to the \TS,
and proves that $\kssfreshcommitment$ contains $\epoch, \hash(\ctxt), \hash(\ctxt_{1}), \ldots, \hash(\ctxt_{\tokensecpar})$ and that the user's tokens have not been revoked.
The \TS decrypts $\ctxt$ and uses the plaintext as the key in the threshold-cryptographic protocol.
But how does the \TS check if $\ctxt$ is correctly formed?
To this end, the user reveals the differences $\kssrandomizerDiff_i = \xsencdelta - \kssrandomizer_i$ for all $i=1,\ldots,\tokensecpar$.
We know from Lemma~\ref{lem:one-good-element} that at least one index $i^*$ exists such that $\ctxt_{i^*}$ is correctly formed.
Therefore, if the differences $\kssrandomizerDiff_i$ are correct,
then because $\ctxt_{i^*}$ is a randomization of $\xs$, so must be $\ctxt$. 
In this, key-share tokens differ from Chaum et al.'s e-cash tokens~\cite{ChaumFN88}, where it suffices that the correct index $i^*$ exists.

\begin{protocol}
  The $\genshares$ protocol is run between an anonymous user and the \TS.
\begin{enumerate}
\item The user takes $(\sku{}, \xsenc, \abcsku)$ and a token $\ksstoken = \ksstokencontent$ as input and connects to the \TS via an anonymous encrypted channel and authenticates the \TS.
\item
  \label{step:genshares:user-prove}
  First, the user retrieves the current revocation list $\revocationlist$ from the \TS. If her public key $\abcpku = \generator^{\abcsku}$ is contained in $\revocationlist$, her tokens are blocked and she aborts the protocol, destroys her tokens, and does not use the \TS again. Otherwise, she sends $\bssignature,\kssfreshcommitment,\epoch,\ctxt, \ctxt_{1},\allowbreak \ldots,\allowbreak \ctxt_{\tokensecpar}$ to the \TS together with a zero-knowledge proof that $\kssfreshcommitment$ commits to these values and that her user's tokens have not been revoked:
  \begin{multline*}
    \textsf{PK}\{ (\abcsku, \kssfreshcommitmentrand) \;:\;
    \generator^{\abcsku} \not\in \revocationlist
    \;\land \\
    \kssfreshcommitment = \commit((\abcsku, \epoch, \hash(\ctxt), \hash(\ctxt_1), \ldots, \hash(\ctxt_{\tokensecpar})), \kssfreshcommitmentrand)
    \}.
  \end{multline*}
  The \TS verifies the proof; that the signature is valid, i.e., $\bsverify(\bspk, \bssignature, \kssfreshcommitment) = \top$; that it has not seen the signature $\bssignature$ before; and that the epoch $\epoch$ corresponds to the current epoch. The \TS aborts if any check fails. The revocation mechanism can be implemented using BLAC (blacklistable anonymous credentials)~\cite{BLACTsangAKS10} or dynamic accumulators~\cite{CamenischKS09}.
\item Next, the user computes $\kssrandomizerDiff_i = \xsencdelta - \kssrandomizer_i > 0$ and $\prandReveal_i = \prand_i^{-1}\cdot \prand$ such that
  \begin{equation}
    \label{eq:gamma-rand-correct}
    \begin{split}
      \ctxt &= \ctxt_i \cdot \penc(\kssrandomizerDiff_i ; \prandReveal_i)
    \end{split}
  \end{equation}
  for $i = 1, \ldots, \tokensecpar$. She sends $\kssrandomizerDiff_1,\allowbreak \ldots,\allowbreak \kssrandomizerDiff_{\tokensecpar},\allowbreak \prandReveal_1, \ldots, \prandReveal_{\tokensecpar}$ to the \TS.
\item The \TS verifies that the $\kssrandomizerDiff_i$s and $\prandReveal_i$s satisfy equation~\eqref{eq:gamma-rand-correct}
  and that $0 < \kssrandomizerDiff_i < 2^{\lengthdelta+1}$. The \TS aborts if any check fails.
\item The \TS decrypts $\ctxt$, and sets $\xsfresh = \pdec(\ctxt) \pmod{\grouporder}$.
\item The user calculates her key share $\xpfresh$ as:
    \begin{equation*}
      \xpfresh \equiv \sku{} - \xsencdelta \pmod{\grouporder}
    \end{equation*}
\end{enumerate}
\end{protocol}
Using Lemma~\ref{lem:one-good-element}, we can show that the decrypted ciphertext $\ctxt$ must also be of the right form.
\confversioncmd{(See the Appendix of the full version~\cite{fullversion} for the proof.)}
\fullversioncmd{(See Appendix~\ref{app:proofs-of-lemmas} for the proof.)}
\begin{lemma}
  \label{lem:all-same}
  If the tuple $(\epoch, \ctxt, \ctxt_1,\allowbreak \ldots, \ctxt_{\tokensecpar},
  )$ with $\kssrandomizerDiff_1, \ldots, \kssrandomizerDiff_{\tokensecpar}$ and $\prandReveal_1, \ldots, \prandReveal_{\tokensecpar}$ satisfies equation~\eqref{eq:gamma-rand-correct}, then with probability $1 - 1/\binom{2\tokensecpar}{\tokensecpar}$ there exists $\xsencdelta < 2^{\lengthdelta+2}$ such that
  \begin{equation*}
    \pdec(\ctxt) = \xs + \xsencdelta
  \end{equation*}
  where $\xs$ is the \TS key-share for the corresponding user.
\end{lemma}

\parait{The range proof in registration is essential.}
The range proof in equation~\eqref{eq:range-proof-xs} in the \registeruser protocol ensures that the plaintext $\xs = \pdec(\xsenc)$ is small compared to the randomizers $\kssrandomizer_i$ and $\xsencdelta$. As a result, the randomized ciphertexts $\ctxt_i$ statistically hide $\xs$. It is not sufficient to skip the range proof and instead choose the randomizers $\kssrandomizer_i$ and $\xsencdelta$ from the full plaintext domain $\NatNumUpTo{\PaillierN}$ to hide $\xs$. Without the range proof, the \TS can construct tokens that it can later recognize by exploiting the fact that a large $\xs$ results in a reduction modulo $\PaillierN$. More precisely, the \TS can set $\xs$ of its target user somewhat large, so that $\xs + \kssrandomizer_{j} > \PaillierN$ (with a non-negligible probability). The user believes that the \TS derives $\xs + \kssrandomizer_{j} \pmod{\grouporder}$ (because she believes no modular reduction took place) and compensates accordingly. However, the \TS actually derives $\xsfresh = (\xs + \kssrandomizer_j \mod{\PaillierN}) \pmod{\grouporder} = \xs + \kssrandomizer_j - (N \mod{\grouporder}).$ To test if the current token is from its target user, the \TS adds $(N \mod{\grouporder})$ to $\xsfresh$. If the guess was correct, the \tcp completes correctly, otherwise the protocol fails. This allows the \TS to detect specific users.

\para{Blocking the Key.}
To block her key, the user runs the \code{BlockShare} protocol with the \TS ensuring that no new key-share tokens are created for her, and that all her unspent tokens are blocked. Note that the user does not require knowledge of $\abcsku$, she only needs some mechanism to authenticate to the \TS.
\begin{protocol} The user authenticates to the \TS. The \TS looks up the user's record $(\xs, \xsenc, \abcpku)$; marks the user as blocked; and adds $\abcpku$ to the revocation list $\revocationlist$ so that the user's unused tokens will be refused.
\end{protocol}

We assume the \TS is honest with respect to blocking, i.e., it correctly blocks
all unspent tokens. A malicious \TS could try to attack privacy by revoking the
tokens of some honest user. Thereby revealing whether the current user is in the
revoked set or not. However, this attack will only work once. As per the first
step of the \genshares protocol, the revoked user will detect this revocation and
then refuses to use the \TS again. Making the list $\revocationlist$ append-only ensures that the \TS will always be caught when it maliciously revokes users.

\subsection{Alternative Constructions}
\label{sec:alternative-tandem}
An alternative method to construct tokens could be to use an authenticated encryption scheme that the user and the \TS evaluate using secure multi-party computation~\cite{Yao86}. The server inputs its key share $\xs$ while the user inputs the randomizer $\kssrandomizer$. The user's output is the authenticated encryption of $\xs + \kssrandomizer$ for the \TS's symmetric key which serves as token. To ensure that the \TS cannot recognize this token, the protocol should resist malicious servers and the circuit should validate the \TS's inputs (i.e., that the encryption key is the same for all users). Similarly, the protocol should resist malicious users to ensure the server's key share does not leak to the user. Such a circuit requires at least 4 block cipher operations and a hash computation. Taking results from recent maliciously secure two-party computation protocols~\cite{WangRK17} shows that this MPC is requires 1 to 2 orders of magnitude more computation and 2 to 3 orders of magnitude more bandwidth than our custom protocol.

Another simple alternative construction is to let users retrieve $\xsenc = \Enc(\xs)$ using private information retrieval (PIR) via an anonymous channel---the user must still hide her identity. Then, users randomize $\xsenc$ similarly to our construction, and the \TS decrypts the ciphertext to recover $\xsenc + \kssrandomizer$, which it then uses in the \textsf{TCP}. To enable blocking of keys, the \TS needs to frequently refresh its encryption keys, effectively invalidating previously retrieved ciphertexts $\xsenc$. This simple protocol, however, has serious drawbacks. First, blocking is only enforced upon key refreshing, thus the time span when compromised keys can be used depends on the refreshing schedule of the \TS. Second, because the encryption of $\xs$ for the current period can be randomized as often as the user wants (and the use of PIR precludes record-keeping), this scheme cannot provide rate-limiting. Third, because the \TS acts as a decryption oracle for a homomorphic encryption scheme, which is only CPA secure, proving security in this setting requires very strong and non-standard assumptions.

 \section{Security and Privacy of \name}
\label{sec:security-privacy}
In this section we formalize the security and privacy properties offered by \name.
We refer to the 
\confversioncmd{full version of this paper~\cite{fullversion}}
\fullversioncmd{appendix}
for the complete security and privacy proofs.

\subsection{Security of \name}
We capture the security of \name using a security game. It models that if the user's key is compromised (e.g., her device is stolen), the user can block the use of her key, assuming the honesty of the \name server.
\begin{game}
  \label{game:security}
  The \emph{\name security game} is between a challenger controlling the \TS and the SP, and an adversary controlling up to $\lastuser$ users. The adversary aims to complete the TCP for a blocked or rate-limited user.
\begin{description}
\item[Setup phase] The challenger sets up the \TS by running $\tandemsetup$. The challenger also sets up the SP. The challenger runs \registeruser with the adversary for each of the $\lastuser$ users the adversary controls.
\item[Query phase] During the query phase, the adversary can ask the \TS to run the  $\obtainkstoken$ and $\blockshare$ protocols with users controlled by the adversary. Moreover, the adversary can make $\qRunTCP$ queries to the challenger. In response, the \TS first runs the \genshares protocol with the user (controlled by the adversary), followed by a run of the \tcp protocol.
\item[Selection phase] At some point the adversary outputs the identifier of a
  blocked or rate-limited user $\U^*$ on which it wants to be challenged later. The challenger runs $\blockshare$ for user $\U^*$ to ensure the user is blocked respectively that the rate-limited user used all tokens.
\item[Second query phase] The adversary can keep asking the \TS to run the
  \obtainkstoken and \blockshare protocols. The adversary can also make $\qRunTCP$ queries as before.
\item[Challenge phase] Finally, upon request of the adversary, the challenger acts as SP in the \tcp protocol. At the same time, the adversary may still make queries and run protocols as before. The adversary wins if it successfully completes the \tcp with the SP on behalf of the blocked user $\U^*$. To prevent trivial wins, this \tcp protocol must be completable only by user $\U^*$ (See Appendix~\ref{app:full-tcps} for how to model this for attribute-based credentials).
\end{description}
\end{game}
In this game, all users are automatically corrupted right from the moment they start the registration protocol. This models the notion that users can even be blocked if an active adversary is present right from the start, and also implies that honest users---which are only corrupted later by an active adversary---can still be blocked.

Of course, to have security using \name, the TCP itself must be secure. Hence, we require that even if a malicious user has interacted many times with the \TS, she cannot use her key when she does not have access to the \TS. We formalize this using the following game.
\begin{game}
  \label{game:tcp-security}
  The \emph{TCP security game} is between a challenger controlling the TS and the SP, and the adversary controlling a malicious user.
  \begin{description}
  \item[Setup phase] During the setup phase, the adversary generates $\xp \randin \Zp$, whereas the \TS, controlled by the challenger, generates $\xs \randin \Zp$.
  \item[Query phase] In the query phase, the adversary can make $\qTCP(\xsencdelta)$ queries to request that the \TS runs $\tcpts(\xs + \xsencdelta)$ with the user. The adversary is responsible for running $\tcpu$. Optionally, the adversary-controlled user can communicate with the challenger-controlled SP running $\tcprp()$ as well.
  \item[Challenge phase] In the challenge phase, the adversary is not allowed to make $\qTCP$ queries. Instead, it interacts solely with the challenger-controlled SP running $\tcprp()$. The adversary wins if the SP accepts.
  \end{description}
\end{game}

\begin{theorem}
  \label{thm:tandem-security}
  No PPT adversary can win the \emph{\name security game} with non-negligible probability, provided that the \tcp is secure (i.e., no PPT adversary can win the TCP security game), the homomorphic encryption scheme is CPA secure, the blind-signature scheme is unforgeable, and the commitment scheme $\extcommit{\cdot}{\cdot}$ is extractable.
\end{theorem}
\begin{proofsketchesinbody}
\begin{proof}[Proof sketch]
  We prove the security of the scheme by reducing it to the TCP security property.
  First, we show how to run \genshares without decrypting ciphertexts. During the \obtainkstoken protocol, we model hash functions as random oracles to allow us to extract the tuple $(\epoch, \hash(\ctxt), \hash(\ctxt_1), \ldots, \hash(\ctxt_{\tokensecpar}))$ contained in the commitment $\ksscommitment$.
  We also extract the corresponding randomizers $\kssrandomizer_i$ and $\prand_i$ from the extractable commitments $\deltacommitment_i$ in step 3. During \genshares we can use the tuple to identify the user (because the attacker cannot forge signatures), and thus the corresponding key share $\xs$. We then find an index $i^*$ such that $\ctxt_{i^*} = \xsenc \cdot \penc(\kssrandomizer_{i^*}, \prand_{i^*})$ using the extracted randomizers $\kssrandomizer_i$ and $\prand_i$. This does not require decrypting $\ctxt_i$. As a result, we know $\ctxt_{i^*}$ decrypts to $\xs + \kssrandomizer_{i^*}$ and therefore $\pdec(\ctxt) = \xs + \kssrandomizer_{i^*} + \kssrandomizerDiff_{i^*}$. By Lemma~\ref{lem:one-good-element} such an index $i^*$ exists with overwhelming probability.

  Knowing the plaintext of $\ctxt$, we no longer need to decrypt ciphertexts to run \genshares. Therefore, we can use the CPA security of the homomorphic encryption scheme to replace the initial ciphertext $\xsenc = \penc(\xs)$ for the challenge user by $\xsenc = \penc(0)$ an encryption of 0. During \genshares we add $\xs$ to compensate. (To enable the reduction to CPA, we simulate the range proof in step 2 of \registeruser.)

  Finally, we answer all queries for the challenge user using the TCP security oracle. Hence, a break of the \name security game results in a break of the TCP security game.
\end{proof}
\end{proofsketchesinbody}

\confversioncmd{See the Appendix of the full version~\cite{fullversion} for the proof.}
\fullversioncmd{See Appendix~\ref{sec:security-proof} for the proof.}

\subsection{Privacy of \name}

The following game models that \name provides key-use privacy: A malicious \name server cannot distinguish between two unblocked honest users performing a transaction using the \TS even if it colludes with the service provider, provided that the service provider alone cannot distinguish transactions by these users. Thus, users are unlinkable when using their keys at the \TS.

\begin{game}
  \label{game:privacy}
  The \emph{\name privacy game with colluding SP} is between a challenger, who controls two honest users $\U_0$ and $\U_1$, and an adversary $\Adv$ who controls the \TS and the SP.
  \begin{description}
  \item[Setup phase] The adversary $\Adv$ outputs the number of key-share tokens $n_0, n_1$ each respective honest user should obtain. The adversary is responsible for setting up the SP and the \TS, i.e., it should publish public keys $\ppk$ and $\bspk$. Next, the honest users $\U_0$ and $\U_1$ run $\registeruser$ with the adversary-controlled \TS and then obtain $n_0$ and $n_1$ key-share tokens respectively. First, $\U_0$ runs $\obtainkstoken$ $n_0$ times to obtain tokens $\ksstoken_{0,1}, \ldots, \ksstoken_{0,n_0}$. Then, $\U_1$ runs $\obtainkstoken$ $n_1$ times to obtain tokens $\ksstoken_{1,1},\ldots,\ksstoken_{1,n_1}$.
\item[Query phase] During the query phase, the adversary can make
  $\qRunTCP(\U_i, j, \auxu)$ queries to request that user $\U_i$ uses token
  $\ksstoken_{i,j}$ and then runs the TCP with input $\auxu$. If $i \in \{0,1\}$ and user $\U_i$ did not use token $\ksstoken_{i,j}$ before, then user $\U_i$, controlled by the challenger, first runs $\genshares$ with the \TS using token $\ksstoken_{i,j}$ and then runs $\tcpu(\auxu)$ with the \TS and the SP (running $\tcpts$ and $\tcprp$ respectively).
\item[Challenge phase] At some point, the adversary outputs a pair of token indices $(i_0,i_1)$ for user $\U_0$ and $\U_1$ respectively on which it wants to be challenged. Let $\ksstoken_0 = \ksstoken_{0,i_0}$ and $\ksstoken_1 = \ksstoken_{1,i_1}$ be the corresponding tokens. The adversary loses if either token $\ksstoken_0$ or $\ksstoken_1$ has been used before or if user $\U_0$ or $\U_1$ detects it is blocked when running $\genshares$. Then, the challenger picks a bit $\challengebit \in \{0, 1\}$ and proceeds as if the adversary made a $\qRunTCP(\U_b, \ksstoken_{b})$ query followed by a $\qRunTCP(\U_{1-b}, \ksstoken_{1-b})$ query.
\item[Guess phase] The adversary outputs a guess $\challengebit'$ of $\challengebit$. The adversary wins if $\challengebit' = \challengebit$.
\end{description}
\end{game}

The privacy game models the fact that there is $\emph{no}$ time correlation between when tokens are obtained by a user, and when they are spent by a user. At the same time, the adversary has full control over the \TS and the SP, so this game also models the fact that the \TS and the SP \emph{can} correlate events that they see.

Since the SP is controlled by the adversary, the \tcp must ensure privacy with respect to the SP and the \TS, if all that the TS sees are randomized secret shares. We formalize this in the following game.
\begin{game}
  \label{game:tcp-privacy}
  The \emph{TCP privacy game with colluding SP} is between a challenger controlling honest users $\User_0$ and $\User_1$ and an adversary $\Adv$, controlling the \TS and the SP.
\begin{description}
\item[Setup] The adversary publishes the \TS public key and is responsible for setting up the SP. The challenger sets up its users. First, user $\User_0$ generates $\sku{0} \randin \Zp$ while the \TS generates $\sks{0} \randin \Zp$, then $\User_1$ and \TS similarly generate $\sku{1}$ and $\sks{1}$. Finally, the \TS sends $\sks{0}$ and $\sks{1}$ to users $\User_0$ and $\User_1$ respectively.
  \item[Queries] Adversary $\Adv$ can make $\qRunTCP(i, \auxu)$ queries, to
    request $\User_i$ to run the TCP protocol using input $\auxu$ with the \TS and the SP (both controlled by $\Adv$). User $\User_i$ picks $\kssrandomizer \randin \Zp$ and sends the randomized secret-share $\xsfresh = \sks{i} + \kssrandomizer \pmod{\grouporder}$ to the \TS. The user sets $\xpfresh = \sku{i} - \kssrandomizer$ and runs $\tcpu(\xpfresh, \auxu)$ with the \TS and the SP running $\tcpts(\xsfresh)$ and $\tcprp$ respectively.
  \item[Challenge] Adversary $\Adv$ outputs an input $\auxu$. The challenger picks a bit $\challengebit \randin \{0, 1\}$. Then the challenger acts as if $\Adv$ first made a $\qRunTCP(b, \auxu)$ query, and then a $\qRunTCP(1 - b, \auxu)$ query.
  \item[Guess] $\Adv$ outputs a guess $\challengebit'$ for $\challengebit$, $\Adv$ wins if $\challengebit = \challengebit'$.
\end{description}
\end{game}
\name also provides key-use privacy against the \TS alone, even if the SP can identify users. (If the SP can identify users, then so can the \TS and the SP together. We exclude this case to prevent a trivial win.) We model this situation as a variant of the previous two games.
\begin{definition}
  \label{def:honest-sp-privacy}
  The \emph{\name privacy game with honest SP} and the \emph{TCP privacy game with honest SP} are as in \cref{game:privacy} and \cref{game:tcp-privacy} above, however, the challenger controls the SP. The adversary can interact with the SP as a normal user.
\end{definition}

\begin{theorem}
  \label{thm:tandem-privacy}
  No PPT adversary can win the \emph{\name privacy game with colluding SP} (respectively the \emph{\name privacy game with honest SP}) with probability non-negligibly better than $1/2$, provided that the \tcp is privacy-friendly (i.e., no PPT adversary can win the TCP privacy game with colluding SP respectively the TCP privacy game with honest SP),
  the commitment scheme $\commit(\cdot, \cdot)$ is computationally hiding,
  and that the commitment scheme $\extcommit{\cdot}{\cdot}$ is extractable.
\end{theorem}

\begin{proofsketchesinbody}
\begin{proof}[Proof sketch]
  We first argue that we can remove all identifying information from the key-share tokens of the challenge users. First, we extract the server's key-shares $\sks{0}$ and $\sks{1}$ from the proof of knowledge in step 2 of the \registeruser protocol. Then we simulate the proof of knowledge in step 6 of \obtainkstoken, replace the ciphertext $\enctokenid$ by the encryption of zero (using the CPA security of the ElGamal encryption scheme), extract the subset $\tokendisclose$ so that we can send random commitments $\ksscommitment_i, \deltacommitment_i$ for $i \not\in \tokendisclose$ (because the commitment schemes are computationally hiding), and finally, we set the unrevealed ciphertexts $\ctxt_i = \penc(\kssrandomizer_i, \prand_i)$ for $i \not\in \tokendisclose$. None of these changes are detectable by the adversary during \obtainkstoken.

  To use tokens as requested by the adversary, the challenge users add $\sks{i}$ to their long-term share $\xp$ to compensate for the changes made so that \genshares completes successfully. Moreover, because the randomizers $\kssrandomizer_i$ statistically hide $\sks{}$, the adversary cannot detect the final change to the ciphertexts during \genshares.

  Therefore, by the blindness of the signature scheme, we can swap tokens between the users and still simulate protocols perfectly. Therefore, any adversary that can then still distinguish users must break the security of the TCP privacy game (with a colluding SP or with a honest SP).
  To extract secrets from the proofs and commitments, and to make the final reduction to TCP privacy, we model hash functions as random oracles.
\end{proof}
\end{proofsketchesinbody}

\fullversioncmd{See Appendix~\ref{sec:privacy-proof} for the proof.}
\confversioncmd{See the Appendix of the full version~\cite{fullversion} for the proof.}
 \section{Securing Protocols with \name} \label{Basic-keysharing}
\name enables the use of a central server in a common class of TCPs without incurring a privacy cost with respect to this server. It operates on specific linearly randomizable TCPs (with corresponding group order $\grouporder$) that satisfy the security an privacy properties in Games~\ref{game:tcp-security} and~\ref{game:tcp-privacy}. A central server that is willing to run these TCPs can trivially offer them in a privacy-preserving manner by implementing \name's protocols in Section~\ref{keyshare-tokens}. As a result, the central server can no longer distinguish users when it runs a TCP. Of course, it can still distinguish different TCPs and see inputs to the TCPs that are not protected by \name.

Many protocols already have threshold-cryptographic versions that satisfy the
necessary properties. For example, threshold variants of Schnorr signatures~\cite{GennaroJKR07} and ElGamal-based encryption schemes~\cite{ElGamal84,ShoupG02} rely on Shamir secret-sharing and are thus linearly randomizable.
A reduction to the original security property shows they also satisfy the TCP security property: given oracle access to a signer or decryptor for a fixed private key, we can answer the $\qTCP(\xsencdelta)$ queries by modifying the original response.
The protocols are also TCP private. The server-side protocols for signing and
decrypting operate solely on the secret-share and the common input, the message
or ciphertext. In Appendix~\ref{app:elgamal-tandem} we show how to protect ElGamal decryption with \name.

\name can enhance threshold-cryptographic versions of electronic cash~\cite{CamenischHL05,MiersG0R13}; group signatures~\cite{BichselCNSW10,CamenischL04}; and attribute-based credentials (ABC)~\cite{BBS+AuSM06,brands2000rethinking,Idemix1-CamenischH02,Idemix2-CamenischL02}. Only Brands' scheme has a threshold-cryptographic version~\cite{brands2000rethinking}. For the others, the threshold-cryptographic versions of the zero-knowledge proofs must be created. As an example, we show how to convert the BBS+ ABC scheme~\cite{BBS+AuSM06} into a \name-suitable threshold-cryptographic scheme. 

Not all threshold schemes are compatible with \name. Threshold DSA signatures are notoriously complicated. Some schemes use multiplicative secret-sharing (instead of additive)~\cite{MacKenzieR04,DoernerKLS18}, and others require additional information such as public keys or encryption keys that break the TCP privacy property~\cite{GennaroGN16,GennaroG18}. Similarly, threshold RSA encryption and decryption~\cite{Shoup00} need the public modulus, and hence do not satisfy the TCP privacy property.

\subsection{Use Case: Attribute-Based Credentials}
\label{sec:bbsplus}

Attribute-based credentials can be conceptualized as digital equivalents to classic identity documents such as passports. 
The owner of a credential can selectively disclose any subset of attributes to a
service provider so that the validity of the disclosed attributes can be verified. In many ABC systems credentials are unlinkable across disclosures, making users anonymous within the set of users having the same disclosed attributes. 

Credentials contain the user's secret key to bind credentials to a user, and to ensure that only the owner can use them. 
\name can be used to strengthen the security of this key to ensures that credentials cannot be abused while preserving users' privacy. 

We now show how to apply \name to BBS+ credentials~\cite{BBS+AuSM06} by converting its issuing and showing protocols into threshold-cryptographic alternatives. BBS+ credentials are anonymous credentials built from BBS+ signatures~\cite{BBS+AuSM06}. BBS+ signatures operate in a pairing setting and rely on discrete-logarithm based assumptions. Let $(\G_1, \G_2)$ be a bilinear group pair, both of prime order $\grouporder$, generated by $g$ and $h$ respectively. The pairing is given by $\hat{e}: \G_1 \times \G_2 \rightarrow \G_T$ where $\G_T$, also of order $\grouporder$, is generated by $\hat{e}(g,h)$. Let $l$ be the number of attributes. In the BBS+ credential scheme, an issuer randomly chooses generators $B, B_0,..,B_l \randin \G_1$, picks a private key $\skI \randin \Zp$, and computes $w=h^{\skI}$. The issuer's public key is $\pkI = (w, B, B_0, .., B_l)$. 

\para{Obtaining a credential.}
Attribute-based credentials contain the user's secret key as an attribute. For simplicity, we describe the \name BBS+ issuance and showing protocols below with two attributes: the secret key $\secret$ and an issuer-determined attribute $a_1$. To obtain a credential, the user and the \name server run the following TCP version of the issuance protocol with the issuer.
\newcommand{\TBBSissuance}{\code{TandemBBS.Issue()}\xspace}
\newcommand{\TBBSshowing}{\code{TandemBBS.Show()}\xspace}
Let $\xpfresh$ and $\xsfresh$ be the shares of the user's secret key $\secret = \xpfresh + \xsfresh$ held by the user and the \TS respectively. The user must commit to her secret key $\secret$ to allow the issuer to blindly sign it.
  As the user's secret key is shared between the user and the \TS, they both have to participate in creating the commitment. First, the user sends $B_0$ to the \TS so that it can compute $B_0^{\xsfresh}$. Then the user and the \TS create a commitment 
  $ U = B^{s'} B_0^{\xpfresh} B_0^{\xsfresh} = B^{s'} B_0^{\secret} $
  where $s' \in_R \Zp$.
To prove to the issuer that $U$ is well-formed, the user and the \TS construct the proof
$ \PK\{(\secret, s'): U = B^{s'} B_0^{\secret}\}. $
The construction of this proof is very similar to the threshold version of Schnorr's protocol in Fig.~\ref{fig:threshold-schnorr}. For completeness, we include the full protocol in
Fig.~\ref{fig:BBS+issuance-commitment-proof-distributed} in Appendix~\ref{sec:abcs-full-tcps}. If this proof of knowledge verifies, the issuer randomly generates  $s'', e \in_R \Zp$, calculates $$ A = \Big( g B^{s''} U B_1^{a_1} \Big)^{\frac{1}{e+\skI}} \quad  \in \G_1, $$ and sends the tuple $(A,e,s'')$ to the user. The user calculates  $s=s' + s''$ and stores the credential $\sigma = (A,e,s)$.

\para{Showing a credential.}
After the issuance protocol, the user can show the credential to authenticate to a service provider. We convert the showing protocol into a TCP that uses the \name server.

Using the showing protocol, the user can prove possession of a credential $\sigma = (A, e, s)$ over her key $\secret$ and a revealed attribute $a_1$ ($\secret$ remains hidden) by convincing the service provider the credential is valid, i.e., that
\begin{equation}
  \label{eq:have-signature}
  \hat{e}(A, h^ew)= \hat{e}(g B^s B_0^{\secret} B_1^{a_1}, h).
\end{equation}
We follow the approach by Au et al.~\cite{BBS+AuSM06} to prove this in zero-knowledge. Let $g_1 , g_2$ be generators in $\G_1$. First, the user creates commitments
$C_1 = Ag_2^{r_1}$ and $C_2 = g_1^{r_1} g_2^{r_2}$ 
for $r_1, r_2 \randin \Zp$, and sends them to the SP.
Finally, she and the \TS engage in the following zero-knowledge proof with the SP:
\begin{multline*}
  \PK \Big\{ (r_1, r_2, \alpha_1, \alpha_2, e, \secret, s) :
      C_2 =   g_1^{r_1} g_2^{r_2} \land C_2^e= g_1^{\alpha_1} g_2^{\alpha_2} \land \\
      \hat{e}(C_1, w) \hat{e}(C_1,h)^e =  \hat{e}(g,h) \hat{e}(B, h)^s
      \mathbf{\hat{e}(B_0, h)^{\secret}} \cdot\mbox{} \\
      \hat{e}(B_1, h)^{a_1} \hat{e}(g_2, w)^{r_1} \hat{e}(g_2, h)^{\alpha_1} \Big\}
\end{multline*}
to prove that she indeed posseses the signature over the hidden and the disclosed attributes and that equation \eqref{eq:have-signature} is satisfied. In the proof, $\alpha_1 = er_1$ and $\alpha_2 = er_2$. The user can herself generate the proofs for the first two conjuncts. The third conjuct, however, contains the user's secret key $\secret$ of which the user only has a share. Thus, the user has to contact the \TS to construct this part of the proof. This proof is just a proof of representation, as before, albeit a bit more complex.
As a result, a construction similar to Fig.~\ref{fig:threshold-schnorr}
and~\ref{fig:BBS+issuance-commitment-proof-distributed} (in the appendix),
allows the user and the \TS to jointly compute this proof. See
Appendix~\ref{sec:abcs-full-tcps} for the full protocol.

\para{Security and privacy of the TCPs.}
These TCPs satisfy the TCP security and privacy notions defined in Section~\ref{sec:security-privacy}.
The \TS computes zero-knowledge proofs of knowing $\xsfresh$. A malicious user learns nothing about $\xsfresh$ (thus nor $\xs$) as a result of the zero-knowledge property.
Hence, the TCP showing and issuance protocols satisfy the TCP security property (see Game~\ref{game:tcp-security}).

For privacy (see Game~\ref{game:tcp-privacy}), the \TS operates on a randomized key $\xsfresh$, so the \TS cannot distinguish users if the SP is honest. The indistinguishability of the credential scheme guarantees that the \TS cannot distinguish by colluding with the SP either.
Thus, the TCP showing protocol is private for honest and colluding SPs.

We refer to Appendix~\ref{app:full-tcps} for the full TCP
security and privacy proofs.

\para{Revocation and rate-limiting with \name.}
When applying \name to ABC schemes, the TS can rate-limit and block keys, and
therefore rate-limit and revoke credentials. Complex custom
solutions create revocable~\cite{PEREAAuTK11,BLACTsangAKS10} and
rate-limitable~\cite{CamenischHKLM06} credentials directly. While the end-result
is similar, \name makes different trust assumptions.

Consider the case where honest users want to protect themselves against
compromise of their credentials. The custom solutions
rely on the honesty of SPs to enforce rate-limits and revocations on
behalf of the user. \name instead ensures rate-limiting and blocking as long as
the TS, which the user chooses, is honest.

The \name approach in this section, however, cannot protect SPs against
malicious users as users are not forced to use \name. In
Appendix~\ref{sec:abc-rate-limiting} we show a modification that forces all users to
use the (same) TS. With this change, the TS can rate-limit and block any
misbehaving users on behalf of issuers and service providers, replacing complex
ad-hoc cryptographic techniques.
 \section{Performance Evaluation} \label{sec:Implementation}

We evaluate \name's computational and bandwidth cost. We use the ABC instantiation as a case study and compare its performance without key protection, with vanilla threshold-cryptographic version of the ABC protocols, and with \name protection.

\name consists of four protocols: \registeruser, \obtainkstoken, \genshares, and
\blockshare. We implemented in C the time-critical protocols, \obtainkstoken and
\genshares.\footnote{Code here: \url{https://github.com/spring-epfl/tandem}} We used Pedersen commitments~\cite{Pedersen91} as commitment scheme, and BBS+ credentials~\cite{BBS+AuSM06} to construct the blind signature (we use an extra attribute containing a serial number to ensure that the blind signature can be used only once). We use the RELIC cryptographic library to implement them~\cite{relic-toolkit}.\footnote{We use a BLS curve over a 381 bits field in RELIC. This setup ensures 128 bits security, while the group order remains 255 bits.} We use a recent implementation~\cite{BarbosaCF17} of Joye and Libert's additive homomorphic encryption scheme~\cite{JoyeL13}, see Appendix~\ref{sec:appendix:joye-libert}.
We set the modulus size to 2048 bits and the size of the plaintext space $\PaillierN$ to 394~bits, such that $\PaillierN > 2^{\lengthdelta+2}$ for $\tokensecpar < 64$ (recall $\lengthdelta$ is the size of the randomizers, see page~\ref{ref:size-lengthdelta}). With this setting, encrypting a single 394 bits plaintext takes 0.9\,ms whereas it takes 24.2\,ms to decrypt a ciphertext.

We empirically measure performance on a single core of an Intel i7-7700 running at 3.6\,GHz.

\para{Obtaining a token.} We first justify our choice for the parameters $\tokensecpar$. Our analysis shows that an attacker can break \name's security property by constructing a key-share token for a blocked user with probability $\binom{2\tokensecpar}{\tokensecpar}^{-1}$. Hence, $\tokensecpar = 42$ gives 80 bits of security, and $\tokensecpar = 66$ gives 128 bits security. However, \obtainkstoken is an \emph{interactive} protocol.
The success probability of an attacker is limited by how often the \TS lets the attacker try to construct a malicious token rather than by the adversary's computational power. As the \TS bans users trying to construct malicious tokens, one can choose a smaller $\tokensecpar$ in practice. In a system with 100\,000 users, $\tokensecpar = 20$ ensures that the probability that an attacker (corrupting all users) can at least \emph{once} use any blocked key is less than $10^{-6}$.

\begin{figure}[tb]
\begin{tikzpicture}
    \begin{axis}[
        xlabel=Difficulty ($\tokensecpar$),
        ylabel=Time (ms),
        legend entries={Obtain User, Obtain Server, GenShares Server},
        legend pos=north west,
        ymin=0,ymax=0.2,
        xmin=0,xmax=70,
        scaled y ticks=base 10:3,
        xlabel near ticks,
        ylabel near ticks,
        legend style={font=\footnotesize},
        label style={font=\footnotesize},
        tick label style={font=\footnotesize},
        ytick scale label code/.code={},
        height=45mm, width=\columnwidth,
      ]
      \pgfplotsset{
        every axis plot post/.style={
          error bars/.cd,
          y dir=both, y explicit,
        },
      }

      \addplot+[black,mark=square] table[x=difficulty,y=ObtainUserTotal] {data/exp-data.dat};
      \addplot+[blue,mark=o] table[x=difficulty,y=ObtainServerTotal] {data/exp-data.dat};
      \addplot+[red,mark=+] table[x=difficulty,y=GenSharesServerTotal] {data/exp-data.dat};

    \end{axis}
  \end{tikzpicture}
  \caption{\obtainkstoken protocol computing time at the user (black) and the server side (blue), and \genshares protocol computation time at the server side (red) for increasing difficulty levels $\tokensecpar$ excluding revocation cost.\label{fig:time-obtain}}
\end{figure}
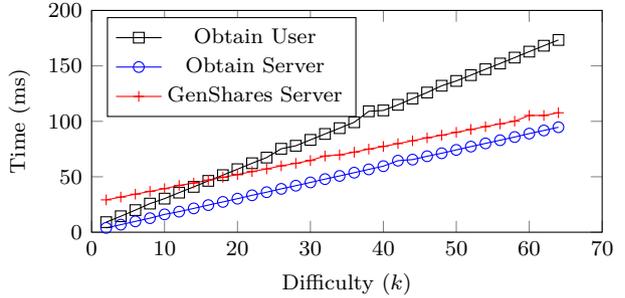

Fig.~\ref{fig:time-obtain} shows the computing time (without communication) for the \obtainkstoken protocol at the user (black) and server (blue) for different values of the parameter $\tokensecpar$. The homomorphic encryption scheme---creating the ciphertexts (user), and checking a subset of these (\TS)---dominates the computational cost. Our experiments reveal that the timing variance across executions is negligible. The bandwidth cost is low: users send and receive $1314\tokensecpar + 405$ and $2\tokensecpar + 117$ bytes respectively. For a security level of $\tokensecpar = 20$, the user sends about 26\,KiB and receives less than 200 bytes.

\para{Using the key.}
We first examine the cost of using a key without the token-revocation check.
On the user side running \genshares is very cheap: less than 5\,ms even for $\tokensecpar = 60$. Users send $594k + 516$ bytes to the server, i.e.,
12\,KiB for $\tokensecpar = 20$ and 35\,KiB for $\tokensecpar = 60$. We show the server's computational cost for recovering the \TS key-share from the token in Fig.~\ref{fig:time-obtain}. For a reasonable security level of $\tokensecpar = 20$, the server computational overhead is around 50\,ms.
The sending of the token in the \genshares protocol can be combined with the request to start the TCP, resulting in no extra latency on top of the delay incurred by the Tor network~\cite{TorDingledineMS04} (1--2\,s to send and receive a small amount of data on a fresh circuit\footnote{As reported by \url{https://metrics.torproject.org/torperf.html}, visited August 31, 2019.}). Circuit creation and $\genshares$ can be run preemptively, thereby reducing the user-perceived delay.

\name uses standard revocation techniques to revoke key-share tokens. Therefore, we did not implement revocation. Tokens expire automatically, so we would only need to block tokens from the current epoch. If the number of blocked users per epoch is small (e.g., for short epochs, such as a day), \name could use BLAC~\cite{BLACTsangAKS10}. BLAC is simple, but proving that a user's tokens are not revoked has linear complexity. Based on the results by Henry and Goldberg~\cite{HenryG13}, we estimate that for 100 blocked users a user needs an additional 20\,ms to prove non-revocation. The \TS needs about 10\,ms to check this proof. If the number of revoked users per epoch is larger, then it is more efficient to use dynamic accumulators~\cite{CamenischKS09}, which are more complex, but constant time. We estimate them to have a 10--20\,ms cost.

Given the above measurements, a modern 4-core server can participate in approximately 50 TCPs per second (not counting the cost of the application-dependent TCP itself), i.e., serve 3\,000 users per minute, requiring about 20 Mbit/s incoming bandwidth.

\para{\registeruser and \blockshare. } These protocols are run rarely (only upon registration and for blocking) and are thus not critical for scalability. We estimate the cost for \registeruser to be well below a second for both the user and the \TS (given its similarity with the \obtainkstoken and \genshares protocols, and that the cost of the range proof is around 500\,ms).
Users have no computational cost when running $\blockshare$. The server's cost is less than a few milliseconds when using BLAC~\cite{HenryG13} or dynamic accumulators~\cite{CamenischKS09}.

\para{Comparison.}
Table~\ref{tab:tandem-overhead} compares the computational cost of creating a single BBS+ showing proof with 5 hidden attributes without key protection, using straight-forward TCP version of the disclosure proof, and using the \name-augmented TCP with $\tokensecpar = 20$.

Without a TCP, the credential showing is very fast and, as there is no other party involved in the use of the key, the showing of the credential is perfectly private. However, it is not possible to perform key blocking nor limit the key-usage without changing the credential type (e.g.~\cite{CamenischHKLM06}) \emph{and} trusting the SP.
The traditional TCP version, has minimal overhead (1\,ms at the server) and provides key blocking and rate limiting, at the cost of privacy. \name provides all three properties. Without taking into account the \obtainkstoken operation that happens offline, the user's overhead is negligible (4\,ms), and well below a second (52\,ms) for the server. In all cases, the delays due to \name's cryptographic operations are small compared to Tor's network delay.

\begin{table}
  \centering
  \caption{\label{tab:tandem-overhead} Comparison of computational cost and properties when running the showing protocol of the BBS+ ABC scheme. We compare not using a TCP, using a traditional TCP, and using a TCP with \name ($\tokensecpar = 20$, excluding revocation cost).}
  \begin{tabular}{lccc}
    & No TCP & Vanilla TCP & TCP + \name \\
    \midrule
    Obtain Token \\
    \enspace User   & -     & -     & 57\,ms \\
    \enspace Server & -     & -     & 30\,ms \\
    Run Protocol \\
    \enspace User    & 5\,ms & 5\,ms & 5 + 4\,ms \\ 
    \enspace Server  & -     & 1\,ms & 1 + 52\,ms \\
    \midrule
    Key blocking &  $\times$   & \checkmark & \checkmark \\
    Rate limiting & $\times$   & \checkmark & \checkmark \\
    Privacy &       \checkmark & $\times$   & \checkmark \\
  \end{tabular}
\end{table}
 \section{Conclusion} \label{sec:Conclusion}

Protecting cryptographic keys is imperative to maintain the security of cryptographic protocols.
As users' devices are often insecure, the community has turned to threshold-cryptographic protocols to strengthen the security of keys.
When run with a central server, however, these protocols raise privacy concerns.
In this paper, we have proposed \name, a provably secure scheme that, when
composed with threshold-cryptographic protocols, provides privacy-preserving
usage of keys. \name also enables users to block their keys and rate-limit their
usage.
Our prototype implementation shows that for reasonable security parameters
\name's protocols run in less than 60\,ms, showing \name's practicality.

\name is particularly suited for privacy-friendly applications such as eCash and ABCs because it retains their inherent privacy properties. Yet, \name can be used to strengthen a wide variety of primitives, including signature and encryption schemes, as long as they can be transformed into linearly-randomizable threshold protocols. Using attribute-based credentials we have shown that deriving such a threshold protocol can be done with standard techniques, and that thereafter adding \name is straightforward.
 \section*{Acknowledgments}
This research is partially funded by the NEXTLEAP
project\footnote{\url{https://nextleap.eu}} within the European Union's Horizon
2020 Framework Program for Research and Innovation (H2020-ICT-2015, ICT-10-2015)
under grant agreement 688722; the Netherlands Organization for Scientific
Research (NWO); and KPN under project ‘Own Your Own Identity (OYOI)’. 
\bibliographystyle{ACM-Reference-Format}
\bibliography{tandem}

\appendix
\section{\name Details and Extensions}

\subsection{The Joye-Libert Encryption Scheme}
\label{sec:appendix:joye-libert}

For completeness, we summarize here the Joye and Libert additively homomorphic
encryption scheme~\cite{JoyeL13} that \name uses. These definitions are
reproduced from the original paper~\cite{JoyeL13}.

\newcommand{\legendre}[2]{\genfrac{(}{)}{}{}{#1}{#2}}
\begin{definition}[\cite{JoyeL13}]
  Let $q$ be an odd prime and let $n \geq 2$ such that $n | (q - 1)$. Then we
  define the \emph{$n$-th power residue symbol modulo q},
  $\legendre{a}{q}_{\!n}$,
  as the smallest (in terms of absolute number) representation of $a^{\frac{q - 1}{n}} \mod{q}$.
\end{definition}

\begin{definition}[\cite{JoyeL13}]
  \newcommand{\jlsize}{\beta}
  \newcommand{\jlmodulus}{n}
  The Joye-Libert additively homomorphic encryption scheme is given by the
  following algorithms.
  \begin{itemize}
  \item $\pkeygen(1^\secpar, \jlsize).$ The key  generation algorithm takes as
    input a security parameter $\secpar$ and
    an integer $\jlsize \geq 1$ indicating the bit-size of the message space.
    Generate primes $r, s \equiv 1
    \pmod{2^{\jlsize}}$, and set the modulus $\jlmodulus = rs$. Finally, pick $y
    \randin \mathbb{J}_{\jlmodulus} \setminus \mathbb{QR}_{\jlmodulus}$, i.e., a
    number whose Jacobi symbol is 1, but that is not a quadratic residue. Return the public
    key $\ppk = (\jlmodulus, y, \jlsize)$ and the private key $\psk = r.$
    The message space will be $\{0, \ldots, 2^{\jlsize} - 1\}$
  \item $\penc(m).$ To encrypt a message $m \in \{0, \ldots, 2^{\jlsize} - 1\}$
    against public key $\ppk = (\jlmodulus, y, \jlsize)$ pick a random
    $r \in \mathbb{Z}_{\jlmodulus}^*$ and return the ciphertext $\ctxt = y^m
    r^{2^{\jlsize}} \mod{\jlmodulus}.$
  \item $\pdec(\ctxt).$ Given a ciphertext $\ctxt$ and a private key $\psk = r$,
    return the message $m \in \{0, \ldots, 2^{\jlsize} - 1\}$ such that
    \begin{equation*}
      \left[ \legendre{y}{r}_{\!\!2^{\jlsize}} \right]^{m} = 
      \legendre{\ctxt}{r}_{\!\!2^{\jlsize}} \pmod{r}.  
    \end{equation*}
    See Joye and Libert~\cite{JoyeL13} for an efficient algorithm to finding $m$. 
  \end{itemize}
\end{definition}

\subsection{Using More Than One \name Server}
\label{sec:appendix-multiple-ts}

\newcommand{\threshold}{t}
\newcommand{\nrparties}{n}

\newcommand{\xss}[1]{x_S^{#1}}
\newcommand{\xssenc}[1]{\overline{x}_S^{#1}}

So far we described \name for use with a single \name server (\TS). In this appendix we show how we can extend $\name$ to support multiple \TSs at the same time. Depending on the configuration, using more than one \TS provides robustness against \TSs that are off-line or deny service to the user, and protects the user against \TSs that violate the security assumptions by colluding with an attacker.

The idea is to use a $\threshold$-out-of-$\nrparties$ Shamir secret-sharing of
the server's secret $\xs$ and to give each \TS one secret-share of $\xs$. Now,
the user needs to cooperate with $\threshold$ \TSs to use her key. Note that the
user's secret $\xp$ is still needed, so no coalition of \TSs can, by themselves,
collude to use the user's key. An attacker that corrupts the user's device, and
thus learns $\xp$, needs to additionally collude with $\threshold$ \TSs to use
the user's key. Therefore, using $\threshold > 1$ strengthens key security if
the security assumption for \TSs cannot be fully guaranteed. If $\threshold <
\nrparties$ some of the \TSs need not participate when the user wants to use her
key, giving the user robustness against denial of service by at most $\nrparties
- \threshold$ non-cooperating \TSs.

The threshold-cryptographic protocols that we have considered in this paper,
i.e., Schnorr signatures, ElGamal encryption and credential schemes, can be
extended to allow Shamir secret sharing of the \TSs keys.

When using multiple \TSs, users must decide which ones compute the TCP correctly.
The user can directly verify the zero-knowledge proofs in for example Schnorr
signatures and the credential protocols (see next section) to check that the \TS
behaved correctly. The ElGamal encryption scheme can be extended so the \TS
proves correct behavior.

We now show how to modify the \name protocols to support multiple \TSs. During \registeruser the user registers with each of the $\nrparties$ \TSs. However, in step~\ref{step:register:choose-secrets}, the $\nrparties$ \TSs jointly compute a $\threshold$-out-of-$\nrparties$ Shamir secret-sharing polynomial $f(X) = a_0 + a_1 X + \ldots + a_{\threshold - 1} X^{\threshold - 1}$ with $a_i \in \Zp$, for example using Pedersen's verifiable secret-sharing protocol~\cite{Pedersen91a}. The server's secret $\xs$ shared by the $\nrparties$ \TSs then equals $a_0$. \name server $i$ stores $\xss{i} = f(i)$ and sends the corresponding ciphertext $\xssenc{i}$ to the user together with the range proof. Each \TS uses its own key-pair for the additively homomorphic encryption scheme.

To obtain a one-time-use key-share token the user runs \obtainkstoken with each of the \TSs in parallel. To support multiple \TSs, the user changes how it chooses value of $\xsencdelta_i$ that it uses for the $i$th \TS in step~\ref{step:obtaintoken:pick-delta}. The user creates a new, random, secret-sharing polynomial $f'(X) = \xsencdelta + b_1 X + \ldots + b_{\threshold - 1} X^{\threshold - 1}$ for $\xsencdelta, b_i \randin \Zp$. Then, the user picks $\xsencdelta_i \randin [2^{\lengthdelta}, 2^{\lengthdelta + 1})$ subject to the constraint that $\xsencdelta_i = f'(i) \pmod{\grouporder}$. Thereafter, the user continues as before.

To use her key, she picks $\threshold$ tokens, and runs \genshares with each of the corresponding \TSs. By choice of the $\xsencdelta_i$ and $\xss{i}$ the \TSs operate on the $\threshold$-out-of-$\nrparties$ shared secret $\xs + \xsencdelta \pmod{\grouporder}$. The user then uses $\xp - \xsencdelta \pmod{\grouporder}$ so that the resulting secret still is $\secret = \xp + \xs.$ \section{Attribute-Based Credentials}
\label{sec:appendix-abcs}

In this section we specify the TCP protocols used for issuance and showing of
BBS+ credentials. Next, we prove that they satisfy the TCP security and privacy
conditions required for \name.

\newcommand{\nonce}{\code{nonce}}
\newcommand{\context}{\code{ctx}}
\begin{figure*}[tb]
  \centering
    \begin{tabular}{lclcl}
      \TS & & User & & Issuer \\
      $\xsfresh \in \Zp, B_0, \context$ & & $\xpfresh \in \Zp, s', g, B, B_0, \context$ & &  $U, \context$ \\
      \midrule
 
      $\serrand \randin \Zp$ & & $\phrand, \srand \randin \Zp$ & &\\ 
      $\sercommit=B_0^{\serrand}$ & & $\phcommit=B_0^{\phrand}$ \\
      $\Ucommit = \Ucommit' \sercommit $ & $\diagramrecv{\Ucommit',\nonce}$ & $\Ucommit' = B^{\srand} \phcommit $ & \diagramrecv{\nonce} & $\nonce \randin \{0,1\}^{128}$ \\
      $\chal = H(\Ucommit \parallel \context \parallel \nonce)$ \\
      $\serResp=\serrand+\chal\cdot\xsfresh$ & $\diagramsend{\Ucommit, \serResp}$ &
               $\chal = H(\Ucommit \parallel \context \parallel \nonce)$ \\
      & & $\sresp = \srand +\chal\cdot s'$ & & \\
      & & $\phResp = \phrand+\chal\cdot\xpfresh$ & & \\
      & & $\xusresp = \phResp+\serResp$ & $\diagramsend{\chal, \sresp, \xusresp}$ & $\Ucommit'' = U^{-\chal}\cdot B^{\sresp}\cdot B_0^{\xusresp}$\\
      & & & & $\chal \? H(\Ucommit'' \parallel \context \parallel \nonce)$\\
      \end{tabular}
  \caption{\label{fig:BBS+issuance-commitment-proof-distributed}Full details of
    the non-interactive proof of knowledge of the user's commitment $ U = B^{s'} B_0^{\xp}
    B_0^{\xs} $ in the BBS+ TCP issuance protocol, where $\context = U
    \parallel B \parallel B_0$ captures the statement to be proven. The \name server only knows $\xsfresh$ and the user knows $\xpfresh$ and the randomness $s'$ (recall $\xsfresh$ and $\xpfresh$ are the respective outputs of the \genshares protocol). The \TS effectively creates a zero-knowledge proof of knowing $\xsfresh$.}
\end{figure*}

\subsection{The Full TCP Protocols}
\label{sec:abcs-full-tcps}
In the body of the paper we argued that the following proof of knowledge
\begin{equation*}
\PK\{(\secret, s'): U = B^{s'} B_0^{\secret}\}.
\end{equation*}
in the BBS+ issuance protocol can be converted to a threshold-cryptographic
version following the ideas of the threshold Schnorr protocol in
Fig.~\ref{fig:threshold-schnorr}. We show the full, non-interactive version of
this protocol in Fig.~\ref{fig:BBS+issuance-commitment-proof-distributed}. The
issuer creates a nonce $\nonce$ to ensure freshness. In these protocols $H:
\{0,1\} \to \Zp$ is a cryptographic hash function which we will later model as a
random oracle.

To enable BBS+ issuance and verification, the TS must participate in two
protocols, one for issuance, shown in
Fig.~\ref{fig:BBS+issuance-commitment-proof-distributed}, and one very similar
protocol for showing credentials. We summarize the TS's side of these protocols.

\begin{protocol}
  \label{tcp:bbsplus}
  The \tcpts protocol run by the TS for BBS+ schemes is as follows. The TS takes
  as input the randomized secret share $\xsfresh$. First the user
  indicates to the TS whether she wants the TS to participate in an issuance
  protocol or a showing protocol. We assume that the group order $\grouporder$
  and the description of the relevant groups are known to the TS.

  \para{Issuance.}
  If the user indicates an issuance protocol, they proceed as follows.
  \begin{enumerate}
    \item The user sends to the TS the generator $B_0$, the
      context $\context = U \parallel B \parallel B_0$, a nonce $\nonce$, and
      the partial commitment $\Ucommit'$.
    \item The TS picks $\serrand \randin \Zp$, and computes the final commitment
      $\Ucommit = \Ucommit' B_0^{\serrand}$. Next, the TS computes the challenge
      $\chal = H(\Ucommit \parallel \context \parallel \nonce)$ and computes its
      response $\serResp = \serrand + \chal \cdot \xsfresh$. The TS sends
      $\Ucommit$ and $\serResp$ to the user.
  \end{enumerate}

  \para{Showing.}
  If the user instead indicates a showing protocol, they proceed as follows.
  Let $E_1 \in \G_1, E_2 \in \G_1, E_3 \in \G_T$ represent the commitments in
  the zero-knowledge proof corresponding to the three respective conjuncts. And let $E_3'$ be the partial commitment without the TS's contribution.
  \begin{enumerate}
    \item The user sends to the TS the generator $\hat{e}(B_0, h)$, the
      context $\context$ representing the bases and the full proof statement, a
      nonce $\nonce$, the commitments $E_1, E_2$ and the partial commitment $\tilde{E_3}'$.
    \item The TS picks $\serrand \randin \Zp$, and computes the final commitment
      $E_3 = E_3' \hat{e}(B_0, h)^{\serrand}$. Next, the TS computes the challenge
      $\chal = H(E_1 \parallel E_2 \parallel E_3 \parallel \context \parallel \nonce)$ and computes its
      response $\serResp = \serrand + \chal \cdot \xsfresh$. The TS sends
      $E_3$ and $\serResp$ to the user.
  \end{enumerate}
\end{protocol}

\subsection{TCP Security and Privacy}
\label{app:full-tcps}

We now prove that the TCP protocol in Protocol~\ref{tcp:bbsplus} is TCP secure
(see Game~\ref{game:tcp-security}). In this case, the SP models the role of
credential issuer and credential verifier. To win in the challenge phase, the
adversary should convince the SP (acting as verifier) that it holds a valid
credential. Of course, this only makes sense if this credential was issued
against a key $\secret = \xp + \xs$ protected by the \TS. Therefore, the SP will
only accept in the challenge phase if it can confirm that it issued this
credential on a \TS protected key $\secret$.

To enable tracking of protected credentials, we assume that the credential
includes a random attribute $a_1 \randin \Zp$ known to the issuer. In most
instances of attribute-based credential schemes such an attribute either already
exists (e.g., a user identifier) or can be cheaply added. The particular choice
of random attribute is not important. What matters is that all protected
credentials have a set of attributes that differentiates them from non-protected
credentials.

To track which credentials are protected, the SP and TS proceed as follows. For
every credential that the SP issues, it checks that the challenge $\chal$ in the
proof of knowledge of the commitment $U$ was
computed by the \TS when running $\tcpts$ in issuance mode. As a result, the
SP can be sure that the TS-protected keyshare $\xsfresh = \xs + \delta$ is included in $U$.
If this is the case, the SP stores the (random) attribute $a_1 \randin \Zp$ for
this credential in a list.

During the challenge phase of the TCP security game, the adversary must demonstrate possession of a
TS-protected credential. To this end, the adversary will reveal the attribute
$a_1$ as part of the showing protocol. The SP will accept, and the adversary
will win, if attribute $a_1$ matches a protected credential.

\begin{theorem}
  Provided that the discrete logarithm problem is hard in $\G_1$ and BBS+
  credentials are unforgeable, the \tcpts protocol in Protocol~\ref{tcp:bbsplus}
  is TCP secure (see Game~\ref{game:tcp-security}) in the random oracle model for $H$.
\end{theorem}
In the proof we show that if the adversary successfully manages to prove
possession of a BBS+ credential with a key that is shared with the TS, then it
must also be able to break discrete logarithms. To reduce to the DL problem, the
challenger will act as if the (unknown) discrete logarithm equals $\xs$. Since the \TS
doesn't know $\xs$, we use the random
oracle to simulate the proofs in which the TS is involved.

In the challenge phase, 
we then extract the
secret, and therefore the discrete logarithm from the successful prover.
Note
that the TS does not participate during the challenge phase, so there are no
conflicts with it simulating the proofs that it then also tries to extract. As a result, we break the discrete logarithm.

\begin{proof}
Suppose that adversary $\Adv$ can convince the SP that it possesses a
credential containing a valid (i.e., protected by the TS) attribute $a_1$. Since
we assumed that forging BBS+ credentials is hard, $\Adv$ must have proven
possession of one of the credentials issued by the TS. We will build an
adversary $\AdvB$ to break the discrete logarithm assumption in $\G_1$.

Let $(G, Z)$ be a discrete logarithm instance in $\G_1$. The goal of $\AdvB$ is to
find $z$ such that $Z = G^z$. To this end, $\AdvB$ sets $B_0 = G$ and proceeds
as if $\xs = z$. Since $\AdvB$ can no longer complete the zero-knowledge
proofs in \tcpts when it receives $\TCP(\delta)$ queries, it simulates them using the random oracle. Let $Y = Z B_0^{\delta}$ be
the randomized public key corresponding to the TS's randomized secret $\xsfresh
= \xs + \delta$ (recall that $\AdvB$ does not know $\xsfresh$).
\begin{enumerate}
  \item For the issuance mode, $\AdvB$ proceeds as follows in step 2. It
    picks $\chal, \serResp \randin \Zp$, sets $\sercommit = B_0^{\serResp}
    Y^{-\chal}$ and computes $\Ucommit = \Ucommit' \sercommit$. Finally, it
    updates the random oracle so that $H(\Ucommit
    \parallel \context \parallel \nonce)$ equals $\chal$.
  \item For the showing mode, $\AdvB$ proceeds similarly. In step 2 it picks
    $\chal, \serResp \randin \Zp$, and computes $E_3 = E_3'
    \hat{e}(B_0^{\serResp}Y^{-\chal}, h).$ Finally, it updates $H$ such that
    $H(E_1 \parallel E_2 \parallel E_3 \parallel \context \parallel \nonce)$
    equals $\chal$.
\end{enumerate}
In both cases, the probability that patching $H$ fails is negligible, because the
input to $H$ is random from the perspective of $\Adv$.

To recover the discrete logarithm of $Z$ we 
must know the user's key $\xp$ such that the credential contains $\secret = \xs
+ \xp$. In fact, nothing limits the user from using a different $\xp$ for each
credential. Therefore, we extract this user-related part as follows.

During the issuance protocol with the SP, the user will interact with the TS to
create the proof of knowledge $(\chal, \sresp, \xusresp)$ as in
Fig.~\ref{fig:BBS+issuance-commitment-proof-distributed} and send it to the SP
acting as issuer. Recall, \AdvB controls SP.
After receiving the proof, \AdvB will rewind $\Adv$ to the
point where it sent $\Ucommit', \nonce$ to the TS. At this point 
$\Adv$'s randomizers $\phrand, \srand$ are fixed. Then \AdvB picks another
challenge $\chal'$ and proceeds as before. $\Adv$ will send another proof
$(\chal', \sresp', \xusresp')$ to the SP. By dividing out the factors the TS
created by simulating its proofs, we end up with two traces $(\Ucommit', \chal,
\sresp, \phResp)$ and $(\Ucommit', \chal', \sresp', \phResp')$ from which \AdvB
extracts the user's secret $\xp$. Therefore, the credential must contain the
secret $\secret = \xp + \xs + \delta$.

Finally, to extract the discrete logarithm of $Z$, $\AdvB$ rewinds $\Adv$ in the
challenge phase, and extracts the secrets, including the secret key $\secret$
encoded in the credential. It looks up the corresponding value $\xp$ extracted
when issuing this credential and returns the discrete logarithm $z = \secret -
\xp - \delta$ of $Z$ with respect to $G$.
\end{proof}

In the full \name security game (Game~\ref{game:security}) we require
additionally that the protocol is only computable by $U^*$ so that the
challenger can confirm a win. This property is trivially satisfied for any
protocol that identifies the user, such as Schnorr’s proof of identification or
signature schemes. In those cases, the SP simply asks for a new proof of
identity or a signature on a new message; and security relates to
the non-impersonation property and non-forgeability properties of the underlying
schemes.

When applying \name to protocols where the user is anonymous with respect to the
SP, e.g., when showing a credential, this condition requires a little bit more
work to verify. Formally, we require that the protocol \tcprp run by the SP in
the challenge phase takes as extra input the identity of a user $U^*$. The
protocol \tcprp will only accept if the SP can verify the identity of user $U^*$.

To facilitate this check for the BBS+ protocols above, we proceed as follows.
First, for every issuance that the adversary participates in, it can choose to
reveal the token owner's identity $U$ to the SP. When the adversary does, the SP
links the random attribute $a_1$ to user $U$. Two, during the challenge phase,
the adversary reveals the attribute $a_1$. If the SP recorded this attribute as corresponding to user $U^*$, the adversary wins.

The TCP privacy property (see Game~\ref{game:tcp-privacy}) follows by
inspection. First, note that the key $\xsfresh$ that users send to the TS
(controlled by the adversary) information theoretically hides the user's
identity. Next, the values that the users send in step 1 of the TCP protocol
(see Protocol~\ref{tcp:bbsplus} above) are independent from any user secret.
Therefore if a coalition of SP and TS can distinguish users, this must be
because of what the user sends to the SP. The theorem follows.

\begin{theorem}
  The \tcpts protocol in Protocol~\ref{tcp:bbsplus}
  is TCP Private (see Game~\ref{game:tcp-security}) against honest SPs.
  Moreover, provided the issuance and showing runs do not identify the user to
  the SP, then these protocols are also TCP private against colluding SPs.
\end{theorem}

\subsection{Rate-limiting in ABCs}
\label{sec:abc-rate-limiting}
Anonymous users can use the cover of privacy to misbehave, negatively impacting the system. ABC systems are not exempt from such misbehavior. Suppose, for example, that a user shares her ``I am older than 18'' credential with many under-aged users who do not hold such a credential. Then, those under-aged users can incorrectly convince service providers that they are over 18 years of age. If this happens often, service providers can no longer rely on these credentials to verify that a user is older than 18. 

To limit such misbehavior, ABCs could benefit from rate-limiting.
One method to limit abuse is to rate-limit credentials by ensuring that
credentials can only be used a limited number of times. For instance, solutions
such as $n$-times anonymous credentials~\cite{CamenischHKLM06} use custom
cryptographic techniques to construct a special type of ABC that can be used
only $n$ times per epoch.

\name can achieve a similar type of rate-limiting \emph{without} modifying the underlying cryptographic construction of ABCs. To rate-limit use of a system, the \TS enforces a per-user and per-epoch limit $q$ on the number of tokens it issues per user and per epoch. As a result, no credential can be shown more than $q$ times per epoch. This approach limits \emph{all} credentials associated to a user's key. If desired, \name can equally be applied on a per-credential basis.

This rate-limiting strategy \emph{requires} that all users use \name. However, recall that the SPs (issuers and verifiers) cannot detect the use of \name, allowing users to forego sharing their keys with the \TS, thus avoiding the rate limit. To enable the \TS to enforce a rate-limit on all credentials, issuers must only issue credentials on keys that are shared with the \TS.

A small change to the threshold-cryptographic version of the issuance protocol
enables the issuer to confirm that users use \name. To signal its involvement,
the \TS signs the challenge $\chal$ and sends the signature $\sigma$ to the
user. The user forwards the challenge to the issuer. The issuer verifies the
signature $\sigma$. If the proofs are correct, then the user's key was shared
with the \TS and the issuer signs the credential.
 \section{Threshold ElGamal Decryption}
\label{app:elgamal-tandem}
As a second example, we show how \name can be easily applied to ElGamal decryption. Let $\secret = \xp + \xs$ be the secret-shared private key and $\schnorrpk = \generator^{\secret}$ the corresponding public key. The ElGamal encryption of a message $m \in \G$ is given by $(c_1, c_2) = (\generator^{r}, m \cdot \schnorrpk^r)$ where $r \randin \Zp$ is an ephemeral key.

To threshold-decrypt a ciphertext $(c_1, c_2)$ the user and the \TS proceed as follows. They first run \genshares, so that the user and the \TS hold the respective shares $\xpfresh$ and $\xsfresh$ such that $\secret = \xpfresh + \xsfresh$. Then the user sends $c_1$ to the \TS which computes $\alpha = c_1^{-\xsfresh}$ and sends it back to the user. The user can now recover the message as $m' = c_2 \alpha c_1^{-\xpfresh}$. Note that the \TS never learns the value of the message. In fact, the user could blind $c_1$ before sending it, to ensure that the \TS cannot recognize the ciphertext either. \section{Proofs of Lemmas}
\label{app:proofs-of-lemmas}

\begin{proof}[Proof of Lemma~\ref{lem:one-good-element}]

  Whenever a ciphertext $\ctxt_i$ is selected by the \TS for opening, the \TS
  checks that it and the corresponding randomizers $\prand_i$,
  $\kssrandomizer_i$, $\deltacommitmentrand_i$, and $\ksscommitmentrand_i$ are as in equation~\eqref{eq:token-element} and that $\kssrandomizer_i < 2^{\lengthdelta}$, and hence as stated in the theorem. 

  Since the \TS checks $\tokensecpar$ tuples, every adversary needs to include at least $\tokensecpar$ correct tuples in its set of $2 \tokensecpar$ tuples. If no index $i^*$ exists for the remaining tuples, then all $\tokensecpar$ of them were incorrectly formed. The probability that none of these $\tokensecpar$ bad tuples were selected during the cut-and-choose protocol is $1 / \binom{2\tokensecpar}{\tokensecpar}.$
\end{proof}

\begin{proof}[Proof of Lemma~\ref{lem:all-same}]
  From Lemma~\ref{lem:one-good-element} we know that with probability $1 - 1/\binom{2\tokensecpar}{\tokensecpar}$ there exists $i^*$ and $\kssrandomizer^*, \xs$ such that 
  \begin{align*}
    \pdec(\ctxt_{i^*}) &= \xs + \kssrandomizer^*
  \end{align*}

  Let $\ctxt = \penc(\alpha)$.
  From equation~\eqref{eq:gamma-rand-correct} we know that:
  \begin{align*}
    \ctxt &= \ctxt_{i^*} \cdot \penc(\gamma_{i^*}; \prand_{i^*})
  \end{align*}
  By decrypting we find that $\alpha = \xs + \kssrandomizer^* + \gamma_{i^*} \pmod{\PaillierN}$. Moreover, $\kssrandomizer^* < 2^{\lengthdelta}$ (by Lemma~\ref{lem:one-good-element}), $\xs < \grouporder < 2^{\lengthdelta}$ (by construction) and $\gamma_{i^*} < 2^{\lengthdelta+1}$ as checked by the \TS. Since $\lengthdelta = \lengthdeltaval$ and $\PaillierN > 2^{\lengthdelta+2}$, we have that $\alpha = \xs + \kssrandomizer^* + \gamma_{i^*}$ as integers, and thus $\ctxt$ is a proper randomization, with randomizer $\kssrandomizer^* + \gamma_{i^*} < 2^{\lengthdelta + 2}$, of $\xs$ as well.
\end{proof}

 \section{Constructing Correctness Proof of $\xsenc$}
\label{sec:appendix-register-proof}

In this section we describe the details of the range proof of $\pdec(\xsenc)$ in the \registeruser protocol. The range proof ensures that the \TS cannot recognize anonymous users by constructing specially crafted versions of $\xsenc$ as explained earlier. When using a homomorphic encryption scheme that supports zero-knowledge proofs, such as Paillier's encryption scheme, we can use standard techniques, see for example the bitwise technique by Bellare and Goldwasser~\cite{BellareG97}, to prove that $\pdec(\xsenc)$ is at most $2\secpar$ bits (which is a sufficient proxy for $\grouporder$ in our schemes).

In our implementation, however, we use Joye and Libert's encryption scheme which
does not readily admit zero-knowledge proofs. Therefore, we instantiate the range proof using a construction that consists of two parts.
{\renewcommand{\theenumi}{\Roman{enumi}}
\begin{enumerate}
\item The \TS constructs a commitment $\ksscommitment$ to $\sks{}$ using a commitment scheme whose message space is at least as big as the plaintext space of the encryption scheme. The \TS then uses a traditional zero-knowledge proof to show that the value $\sks{}$ committed in $\ksscommitment$ is smaller than $\grouporder$.
\item Next, the \TS uses a cut-and-choose technique to show that $\ksscommitment$ commits to $\pdec(\xsenc) = \sks{}$. 
\end{enumerate}}

The details are as follows. The user and \TS take $\xsenc$ as input. The \TS takes as private input $\sks{}$ and the randomizer $\prand$ used to construct $\xsenc$. Let $\bigG$ be a cyclic group of order $\biggrouporder$ generated by $\biggenerator$ such that $\biggrouporder > \PaillierN$ (recall, $\PaillierN$ is the size of the plaintext domain of the homomorphic encryption scheme). Let $\biggeneratorh$ be another generator of $\bigG$ such that the discrete logarithm of $\biggeneratorh$ with respect to $\biggenerator$ is unknown. We use this group to create a commitment scheme with a large message space.

The full protocol has 7 steps. Part I is represented by step 1, whereas part II is represented by the cut-and-choose technique in steps 2 -- 7. If at any step a verification fails, the protocol is aborted. The cut-and-choose technique is very similar to the construction we use in the \obtainkstoken and \genshares protocols. Let $\tokensecpar$ be the difficulty level of the cut-and-choose protocol.

\begin{enumerate}
\item
  \label{step:rangeproof:rangeproof}
  The \TS picks $\ksscommitmentrand \randin \bigZp$, and computes the commitment $\ksscommitment = \biggenerator^{\sks{}} \biggeneratorh^{\ksscommitmentrand}$. Next, the \TS creates a non-interactive proof that the commitment
  $\ksscommitment$ contains key-share $\sks{}$ of the correct size:
    \begin{equation}
      \label{eq:proof-range-xs}
    PK\{ (\sks{}, \ksscommitmentrand) : \ksscommitment = \biggenerator^{\sks{}} \biggeneratorh^{\ksscommitmentrand} \land
    \sks{} \in [0, \grouporder)\},
  \end{equation}
  and sends $\ksscommitment$ and this proof to the user. This proof can be implemented using a standard technique like the bitwise commitment technique of Bellare and Goldwasser~\cite{BellareG97}. The user checks the correctness of the proof.
\item
  \label{step:rangeproof:user-commit}
  The user randomly chooses a subset $\tokendisclose \subset \{1, \ldots, 2\tokensecpar\}$ of cardinality $\tokensecpar$. She commits to $\tokendisclose$ by picking $\discloseSetCommitmentRand \randin \{0,1\}^{\secpar}$ and sending $\discloseSetCommitment = \extcommit{\tokendisclose}{\discloseSetCommitmentRand}$ to the \TS.
\item
  \label{step:rangeproof:ts-commit}
  The \TS picks randomizers $\kssrandomizer_1, \ldots, \kssrandomizer_{2\tokensecpar} \randin
  \{0,1\}^{\lengthdelta}$ and $\prand_1, \ldots, \prand_{2\tokensecpar} \randin \prandspace$
  to construct ciphertexts, and $\ksscommitmentrand_1, \ldots, \ksscommitmentrand_{2\tokensecpar} \in \bigZp$ to create commitments. Then, the \TS sets:
  \begin{equation}
    \begin{split}
    \ctxt_i      &= \penc(\kssrandomizer_i; \prand_i) \\
    \ksscommitment_i &= \biggenerator^{\kssrandomizer_i} \biggeneratorh^{\ksscommitmentrand_i}
    \end{split}
    \label{eq:generate-token-element}
  \end{equation}
  for $i = 1, \ldots, 2\tokensecpar$. Finally, the \TS sends the ciphertexts $\ctxt_1,\allowbreak \ldots,\allowbreak \ctxt_{2\tokensecpar}$ and commitments $\ksscommitment_1, \ldots, \ksscommitment_{2\tokensecpar}$ to the user. The commitments are computationally binding and information theoretically hiding. (Contrary to the \code{ObtainKeyShareToken} protocol, the \TS can safely send the ciphertexts, because the user cannot decrypt them.)
\item
  \label{step:rangeproof:user-open}
  The user sends the subset $\tokendisclose$ and the commitment randomizer $\discloseSetCommitmentRand$ to the \TS.
\item
  \label{step:rangeproof:ts-open}
  If $\discloseSetCommitment = \extcommit{\tokendisclose}{\discloseSetCommitmentRand}$, then the \TS sends $(\kssrandomizer_i, \prand_i, \ksscommitmentrand_i)_{i \in \tokendisclose}$ to the user (otherwise, it aborts). The user verifies that the values $\ctxt_i, \ksscommitment_i$ for $i \in \tokendisclose$ satisfy equation~\eqref{eq:generate-token-element}. Moreover, the user checks that $\kssrandomizer_i < 2^{\lengthdelta}$ for $i \in \tokendisclose$.
\item
   \label{step:rangeproof:ts-relate}
  Next, the \TS computes
    \begin{equation*}
      \kssrandomizerReveal_i = \kssrandomizer_i - \sks{}, \quad
      \ksscommitmentrandReveal_i = \ksscommitmentrand_i - \ksscommitmentrand, \quad
      \prandReveal_i = \prand_i \prand^{-1}
    \end{equation*}
    for $i \not\in \tokendisclose$, and sends them to the user.
  \item
    \label{step:rangeproof:user-verify}
    Finally, the user checks that
    \begin{equation}
      \label{eq:keygen-user-check}
      \begin{split}
      \ctxt_i &= \xsenc \cdot \penc(\kssrandomizerReveal_i; \prandReveal_i) \\
      \ksscommitment_i &= \ksscommitment \cdot \biggenerator^{\kssrandomizerReveal_i} \biggeneratorh^{\ksscommitmentrandReveal_i}
      \end{split}
    \end{equation}
    and that $0 \leq \kssrandomizerReveal_i < 2^{\lengthdelta}$ for $i \not\in\tokendisclose$, and accepts the proof if all verifications are correct.
\end{enumerate}

\begin{lemma}
  \label{lem:ctxt-pk-correct}
  If the user does not reject in the above protocol, then with probability $1 - 1/\binom{2\tokensecpar}{\tokensecpar}$ we have that $\pdec(\xsenc) \in [0,\grouporder)$ as required.
\end{lemma}
\begin{proof}
  From the zero-knowledge proof in step 1, we know that the \TS knows an opening $\alpha', \ksscommitmentrand'$ of $\ksscommitment = \biggenerator^{\alpha'}\biggeneratorh^{\ksscommitmentrand'}$
  such that $0 \leq \alpha' < \grouporder$. We complete the proof by showing that $\alpha' = \pdec{(\xsenc)}$.

  We continue as per Lemma~\ref{lem:one-good-element} and Lemma~\ref{lem:all-same}. We restate them here for completeness. First, along the lines of Lemma~\ref{lem:one-good-element}, with probability $1 - 1/\binom{2\tokensecpar}{\tokensecpar}$ there exists an index $i^*$ such that the \TS knows an opening $\kssrandomizer^*, \ksscommitmentrand^*$ such that:
  \begin{equation}
    \label{eq:point-exists}
    \begin{split}
      \kssrandomizer^* &= \pdec(\ctxt_{i^*}) < 2^{\lengthdelta} \\
      \ksscommitment_{i^*} &= \biggenerator^{\kssrandomizer^*} \biggeneratorh^{\ksscommitmentrand^*}.
    \end{split}
  \end{equation}
  The user checks that the \TS knows an opening for the $\tokensecpar$ pairs that are opened by the \TS in step 4. So, the \TS must include at least $\tokensecpar$ pairs for which it knows a correct opening. Suppose, for contradiction, that the index $i^*$ does not exist, i.e., that the remaining $\tokensecpar$ pairs are incorrect or cannot be opened by the \TS. Since the protocol completed, the user did not detect foul play. This situation can only occur if the \TS correctly guesses the set $\tokendisclose$ in advance. Since the \TS does not learn anything about $\tokendisclose$ before step 3, the probability that none of the remaining pairs is correct is $1 / \binom{2\tokensecpar}{\tokensecpar}$, as required.
  
  Assume now that this index $i^*$ as required above exists. We use this to show that $\ksscommitment$ commits to $\pdec(\ctxt)$, i.e., that $\alpha' = \pdec(\ctxt)$. From equation~\eqref{eq:keygen-user-check} we know that:
  \begin{equation*}
    \begin{split}
      \ksscommitment_{i^*} &= \ksscommitment \cdot \biggenerator^{\kssrandomizerReveal_{i^*}} \biggeneratorh^{\ksscommitmentrandReveal_{i^*}}
    \end{split}
  \end{equation*}
  so, by using equation~\eqref{eq:point-exists} and equating exponents, we find that $\kssrandomizer^* = \alpha' + \kssrandomizerReveal_{i^*} \pmod{\biggrouporder}$.
  We know from the zero-knowledge proof that $\alpha' < \grouporder$ and by direct inspection that $\kssrandomizerReveal < 2^{\lengthdelta}$ therefore, the equality holds over the integers as well, and we have
  \begin{equation}
    \label{eq:delta-one}
    \kssrandomizer^* = \alpha' + \kssrandomizerReveal_{i^*} < 2^{\lengthdelta + 1} < \PaillierN.
  \end{equation} 
  From equation~\eqref{eq:keygen-user-check} we also know that:
  \begin{equation*}
    \ctxt_{i^*} = \xsenc \cdot \penc(\kssrandomizerReveal_{i^*}; \prandReveal_{i^*})
  \end{equation*}
  By decrypting and using equation~\eqref{eq:point-exists} we find that:
  \begin{equation*}
    \kssrandomizer^* = \pdec( \xsenc \cdot \penc(\kssrandomizerReveal_{i^*}; \prandReveal_{i^*}) ) = \pdec(\xsenc) + \kssrandomizerReveal_{i^*} \pmod{\pailliermodulus}.
  \end{equation*}
  Substituting $\kssrandomizer^*$ from equation~\eqref{eq:delta-one} and substracting $\kssrandomizerReveal_{i^*}$ shows that $\alpha' = \pdec(\xsenc) \pmod{\PaillierN}$, and therefore, by size of $\alpha'$ and $\pdec(\xsenc) < \PaillierN$, that $\alpha' = \pdec(\xsenc)$ as required.
\end{proof}

In the security proof, we replace $\xsenc$ with the encryption of 0, so that the adversary who has corrupted a user learns nothing about $\sks{}$ (except what is revealed as a result of the threshold-cryptographic protocol). The following lemma states that we can do so, without the adversary detecting this change.
\begin{lemma}
  \label{lem:can-replace-ctxt}
  TS can simulate the correctness proof given above such that $\xsenc = \penc(0)$, provided that the encryption scheme is CPA secure and the commitment scheme $\extcommit{\cdot}{\cdot}$ is extractable. This simulation does not require any knowledge of how $\xsenc$ was created.
\end{lemma}
This proof uses a sequence of games that interpolates between the situation where the \registeruser protocol is executed normally, and the situation, where $\xsenc$ is an encryption of 0. This game is as in the security game: the adversary can make \registeruser, \obtainkstoken, \genshares, and \blockshare queries. It's task is to determine if $\xsenc$ is as in the original protocol, or $\xsenc = \penc(0)$. In particular:
\begin{itemize}
  \item Game 0. In Game 0, $\xsenc$ is constructed as per the protocol.
  \item Game 1. We proceed as in Game 0, but simulate the cut-and-choose proof in steps~\ref{step:rangeproof:user-commit} --~\ref{step:rangeproof:user-verify} by extracting $\tokendisclose$.
  \item Game 2. As in Game 1, but simulate the zero-knowledge proof in step~\ref{step:rangeproof:rangeproof} of the protocol.
  \item Game 3. As in Game 2, but replace the commitment $\ksscommitment$ by a random commitment.
  \item Game 4. As in Game 3, but replace $\xsenc$ with an encryption of 0.
\end{itemize}
We show that each pair of consecutive games is indistinguishable to a polynomial-time adversary. Hence, no adversary can distinguish Game 0 from Game 4, thus proving the lemma.

\begin{proof}[Proof of Lemma~\ref{lem:can-replace-ctxt}]
  We first show how to simulate the cut-and-choose proof in steps~\ref{step:rangeproof:user-commit} --~\ref{step:rangeproof:user-verify}. The adverary sends a commitment $\discloseSetCommitment$ to the \TS in step 1. We use the extractability of $\extcommit{\cdot}{\cdot}$ to recover $\tokendisclose$ from $\discloseSetCommitment$ (for example, using the random oracle model if it is implemented using a hash-function). 

  We change how \TS acts in step~\ref{step:rangeproof:ts-commit}. Let $\tokendisclose \subset \{1, \ldots, 2\tokensecpar \}$ be the subset of cardinality $\tokensecpar$ extracted from $\discloseSetCommitment$. For all $i \in \tokendisclose$ the \TS sets $\ctxt_i$ and $\ksscommitment_i$ as per equation~\eqref{eq:generate-token-element}. For other elements, i.e., for $i \in \{1, \ldots, 2\tokensecpar\} \setminus \tokendisclose$, the \TS generates $\kssrandomizerReveal \randin \{0, \ldots, 2^{\lengthdelta} - 1\}, \ksscommitmentrandReveal \randin \bigZp, \prandReveal \randin \prandspace$ and sets $\ctxt_i$ and $\ksscommitment_i$ as per equation~\eqref{eq:keygen-user-check}.

  In step~\ref{step:rangeproof:user-open}, the adversary reveals $\tokendisclose$ and $\discloseSetCommitmentRand$. If $\discloseSetCommitment = \extcommit{\tokendisclose}{\discloseSetCommitmentRand}$ then with overwhelming probability, we correctly extracted $\tokendisclose$.
  If we correctly extracted $\tokendisclose$, the \TS can open the tuples for $i \in \tokendisclose$ in step 5 and return $\kssrandomizerReveal_i, \ksscommitmentrandReveal_i, \prandReveal_i$ for the other elements. Both satisfy the adversary's checks in steps~\ref{step:rangeproof:ts-open} and~\ref{step:rangeproof:user-verify}.

  Game 0 is indistinguishable from Game 1. The simulated proof can go wrong for two reasons. One, we can fail to extract the disclose set $\tokendisclose$, but this can only happen with negligible probability. Second, the distribution of $\kssrandomizerReveal_i$s for $i \not\in \tokendisclose$ is not completely correct, however, the size of $\kssrandomizer$ ensures that this difference is statistically hidden from the adversary. So, from the point of view of the adversary, Games 0 and 1 are indistinguishable.

  In Game 2 we simulate the zero-knowledge proof in step 1. By construction of the simulator of this proof, the adversary cannot detect this change.

  As a result of the changes we made in Game 1, the answers of \TS do not depend on the opening of $\ksscommitment$. So, in Game 3 the \TS can generate a random commitment $\ksscommitment \randin \bigG$. Since Pedersen's commitment scheme is information-theoretically hiding, the adversary cannot detect this change.

  In Game 4, the \TS sends $\xsenc = \penc(0)$ to the user instead of an encryption of the key-share $\sks{}$. As a result of the changes we made in Game 1, the \TS can still complete the remaining part of the protocol.

  We claim that the adversary $\Adv$ cannot distinguish Games 3 and 4. Suppose to the contrary that $\Adv$ \emph{can} distinguish Games 3 and 4. We then show that $\Adv$ can break the CPA security of the homomorphic encryption scheme.

  To do so, we build an adversary \AdvB against the CPA security of the encryption scheme. Recall that \AdvB can make a challenge query on two messages $m_0$ and $m_1$. In our case, \AdvB picks $m_0 = \sks{}$ and $m_1 = 0$. Then, its challenger returns a ciphertext $\ctxt_{*} = \penc(\ppk, m_{\challengebit})$ for some bit $\challengebit \randin \{0, 1\}$. Adversary \AdvB needs to guess $\challengebit$. 

  In $\registeruser$ queries for the challenge user $\U^*$, adversary \AdvB (which acts as a challenger to $\Adv$) uses $\xsenc = \ctxt_*$. Clearly, if $\challengebit = 0$, then $\AdvB$ perfectly simulates Game 3. If $\challengebit = 1$, it perfectly simulates Game 4. Therefore, if $\Adv$ can distinguish between Games 3 and 4, it can break the CPA security of the encryption scheme.
\end{proof} \section{Security Proof}
\label{sec:security-proof}

In the security proof, the challenger controls the \TS and the adversary tries to attack a user.
The security proof is a sequence of games. In the final game, the challenger simulates the game using only
the TCP oracle of the TCP security game, without knowing the corresponding \TS' key share $\sks{}$. As a result, any adversary that manages to use the blocked key of that user (or uses that key more often than the rate-limit allows in this epoch) must therefore break the security of the underlying threshold-cryptographic protocol.

We use the following sequence of games:
\begin{itemize}
\item Game 0: We play the game as described in the \name Security game, see  page~\pageref{game:security}.
\item Game 1: We change the definition of $\genshares$. The challenger simulates the workings of \TS but does not decrypt any ciphertext. Instead, the \TS uses the extractability of $\extcommit{\cdot}{\cdot}$
  and the $\deltacommitment_i$s (from the corresponding \obtainkstoken protocol) to compute
  the plaintext corresponding to $\ctxt$ (without decrypting), which it uses as $\xsfresh$.
  The \TS constructs the range proof of $\xsenc$ in the $\registeruser$ protocol as before.
\item Game 2: We guess the challenge user $\U^*$ and we change the definition of $\registeruser$ for this user: we replace $\xsenc = \penc(\xs)$ by $\xsenc = \penc(0)$.
\item Game 3: For all non-challenge users we answer $\genshares$ queries as in the previous game. For $\U^*$ the \TS simulates the TCP following $\genshares$ using the TCP security oracle (without knowing $\xs$ of $\U^*$).
\end{itemize}

We then prove the following:
\begin{itemize}
\item The adversary cannot distinguish Game 0 from Game 1. We prove that as long as one of the pairs $(\ctxt_i, \deltacommitment_i)$ is as it should be---and Lemma~\ref{lem:one-good-element} shows that this is the case with high probability---then we correctly recover the plaintext of $\ctxt$ and thus the TS extracts the correct $\xsfresh$, and therefore the TCP is correct as well.
\item The adversary cannot distinguish Game 1 from Game 2. We no longer decrypt ciphertexts. Hence, we can use the CPA security of the encryption scheme to show that the adversary cannot distinguish Game 1 from Game 2. More formally, we build a distinguisher that interpolates between Games 1 and 2. The distinguisher makes a query for  $m_0 = \xs$ and $m_1 = 0$ to its CPA challenger, and uses the answer as $\xsenc$. Lemma~\ref{lem:can-replace-ctxt} shows the adversary cannot detect this change to $\registeruser$. If the CPA challenger returned an encryption of $\xs$ then the distinguisher perfectly simulates Game 1, otherwise it simulates Game 2. We can still answer $\genshares$ queries correctly, since we no longer need to decrypt any ciphertexts.
\item The adversary cannot distinguish Game 2 from Game 3 because the TCP oracle simulates the same protocol.
\item Finally, if we have an adversary that can win Game 3, then it breaks the security of the TCP because by construction the challenger has no unrevoked respectively unused tokens in the challenge phase for the challenge user $\U^*$ because the user is blocked respectively rate-limited.
\end{itemize}

\begin{proof}[Proof of Theorem~\ref{thm:tandem-security}]
This proof follows the sequence of games highlighted above. Let $\U^*$ be the challenge user. We guess this user. If the guess turns out to be incorrect, we repeat the reduction with a new guess.

In Game 1 we change how the \TS responds to $\qRunTCP$ queries, in particular, we change $\genshares$ for the challenge user $\U^*$. The \TS (controlled by the challenger) no longer decrypts the ciphertext $\ctxt$ revealed in a token, but instead directly recovers the
plaintext using the $\deltacommitment_i$ and $\kssrandomizerDiff_i$ values. The \TS then
continues as before.

To enable the \TS to answer $\qRunTCP$ queries without decrypting, the \TS stores some extra values whenever $\Adv$ runs the $\obtainkstoken$ protocol. Whenever the \TS issues a credential $\cred$, it extracts the attributes $(\epoch, \abcsku, \hash(\ctxt), \hash(\ctxt_1), \ldots, \hash(\ctxt_{\tokensecpar}) )$ (normally, the \TS cannot learn these values). The challenger uses the extractability of $\extcommit{\cdot}{\cdot}$ to find inputs $\kssrandomizer_{i_1}', \ldots, \kssrandomizer_{i_m}'$ and $\prand_{i_1}', \ldots, \prand_{i_m}'$ used to create the unopened commitments $\deltacommitment_{i_1},\ldots,\deltacommitment_{i_m}.$ (The adversary might cheat so that not all $\deltacommitment_i$s are true commitments.) By Lemma~\ref{lem:one-good-element}, $m \geq 1$, and there exists $i^*$ such that the extracted inputs are correct, i.e., $\kssrandomizer_{i^*}' = \kssrandomizer_{i^*}$ and $\prand_{i^*}' = \prand_{i^*}$. The challenger records the tuple $(\U,
(i_1, \hash(\ctxt_{i_1}), \kssrandomizer_{i_1}', \prand_{i_1}'), \ldots, (i_m, \hash(\ctxt_{i_m}), \kssrandomizer_{i_m}', \prand_{i_m}'))$ for later use.

We now show how to answer \qRunTCP queries without needing to decrypt the ciphertexts. The \TS initially follows the $\genshares$ protocol. At the start of the protocol, $\Adv$ proves possession of a fresh, unrevoked signature on the values $(\epoch, \hash(\ctxt), \hash(\ctxt_1), \ldots, \hash(\ctxt_{\tokensecpar}))$
to the \TS (run by the challenger). Moreover, $\Adv$ provides $\kssrandomizerDiff_1, \ldots, \kssrandomizerDiff_{\tokensecpar}$ and $\prandReveal_1, \ldots, \prandReveal_{\tokensecpar}$. The \TS then checks that these values are correct. If not, it aborts. So far, the challenger follows the protocol.

Now, we deviate from the protocol. By the unforgeability of the blind signatures, this signature must have been obtained by running $\obtainkstoken$. Hence, the challenger can look up the corresponding tuple $(\U,\allowbreak (i_1, \hash(\ctxt_{i_1}), \kssrandomizer_{i_1}', \prand_{i_1}'),\allowbreak \ldots,\allowbreak (i_m, \hash(\ctxt_{i_m}), \kssrandomizer_{i_m}', \prand_{i_m}'))$ from tokens it signed by matching on the hashed ciphertexts. Let $\xsenc$ be the encrypted key share for this user. We use the values $\kssrandomizer_{i_1}',\ldots,\kssrandomizer_{i_m}'$ and $\prand_{i_1}',\ldots,\prand_{i_m}'$ to find the plaintext of one of $\ctxt_{i_1},\ldots,\ctxt_{i_m}$ and then use this to compute the plaintext of $\ctxt$.

For $i \in i_1, \ldots, i_m$ test if:
\begin{equation*}
  \ctxt_i = \xsenc \cdot \penc(\kssrandomizer_i', \prand_i')
\end{equation*}
Let $(i^*, \kssrandomizer_{i^*}', \prand_{i^*}')$ be the tuple that satisfies this equation. By Lemma~\ref{lem:one-good-element} we know that there must exist an index $i^*$ such that:
\begin{align*}
  \ctxt_{i^*} &= \xsenc \cdot \penc(\kssrandomizer_{i^*}, \prand_{i^*}), \\
  \deltacommitment_{i^*} &= \extcommit{(\kssrandomizer_{i^*}, \prand_{i^*})}{\deltacommitmentrand^*},
\end{align*}
so this procedure does indeed find such a tuple $(i^*, \kssrandomizer_{i^*}', \prand_{i^*}')$. The plaintext of $\ctxt_{i^*}$ thus is $\sks{} + \kssrandomizer_{i^*}'$. Therefore, the plaintext of $\ctxt$ is $\sks{} + \kssrandomizer_{i^*}' + \kssrandomizerDiff_{i^*}$ because $\ctxt = \ctxt_{i^*} \cdot \penc(\kssrandomizerDiff_{i^*}, \prandReveal_{i^*})$. Therefore $\xsfresh = \sks{} + \kssrandomizer_{i^*}' + \kssrandomizerDiff_{i^*} \pmod{\grouporder}$.

Now that the challenger has derived $\xsfresh$ it continues with the TCP as normal. This shows how we can answer $\qRunTCP$ queries without needing to decrypt the ciphertexts.

Games 0 and 1 cannot be distinguished by the adversary. During $\code{ObtainKeyShareToken}$ queries, the \TS extracts attributes using rewinding, so this is not detected by the adversary. By Lemma~\ref{lem:one-good-element} the index $i^*$ exists with overwhelming probability, so the responses of the \TS are completely identical for the \qRunTCP queries made by the adversary.

Let $\sks{}^*$ be the \TS' key-share for the challenge user $\U^*$. In Game 2, we do not send $\xsenc = \penc(\sks{}^*)$ to the adversary when it makes $\registeruser$ queries for the challenge user $\U^*$. Instead, we send $\xsenc = \penc(0)$.
During $\qRunTCP$ queries, we first extract the plaintext of $\ctxt$ as above, and then add $\sks{}^*$.
The fact that the \TS does not need to decrypt $\ctxt$ to answer \qRunTCP queries together with Lemma~\ref{lem:can-replace-ctxt} shows that the adversary cannot detect this change.

In Game 3, we again change how we answer \qRunTCP queries for the challenge user $\U^*$.  In particular, we will answer this query without using the corresponding key-share $\sks{}^*$. Instead, we use the challenge oracle for the TCP security in the query phase. We proceed as before, to find the plaintext $\xsencdelta$ of $\ctxt$ when running \genshares. However, now we use the TCP challenge oracle to run the TCP by making a $\tcp(\delta \mod{\grouporder})$ query. The \name security challenger relays the messages to the adversary $\Adv$. After the selection phase, we advance the TCP security challenger to the challenge phase. Moreover, the challenge user $\U^*$ cannot obtain new tokens (because $\U^*$ is either blocked or rate-limited), and all old tokens have been revoked or used, so we no longer need access to the TCP oracle to answer queries. Finally, in the challenge phase, we relay the messages to the TCP challenger. Then, if adversary $\Adv$ wins Game 3, it breaks the TCP security of the underlying TCP. Since we assumed this cannot happen, the \name scheme is secure as well. The only difference between Game 2 and Game 3 is that we use the TCP oracle to run the TCP. However, since the TCP oracle uses to correct randomized key, this change is indistinguishable to the adversary.
\end{proof}

 \section{Privacy Proof}
\label{sec:privacy-proof}

In our privacy proof, we reduce an attacker against privacy to an attacker on the blind singing property of the blind signature scheme. We use the following version based on the one by Baldimtsi and Lysyanskaya~\cite{BaldimtsiL13}. 

\begin{game} The \emph{blind signature game} is between a challenger controlling an honest user $\U$ and an adversary controlling the issuer and the verifier.
\begin{description}
\item[Setup] At the start of the game, $\Adv$ publishes the public key $\bspk$ of the signer and outputs all other necessary public parameters.
\item[Challenge] At some point, the adversary outputs two tuples of attributes $m_0$ and $m_1$
  on which it wants to be challenged. The challenger picks a bit $\challengebit \randin \{0, 1\}$ and randomizers $\ksscommitmentrand_0, \ksscommitmentrand_1$, and computes two commitments
  \begin{equation*}
    \begin{split}
      \ksscommitment_0 &= \commit(m_0, \ksscommitmentrand_0) \\
      \ksscommitment_1 &= \commit(m_1, \ksscommitmentrand_1)
    \end{split}
  \end{equation*}
  and proceeds as follows. First, it runs
  the $\bssign(\bspk, \ksscommitment_b)$ protocol
  with private input $m_b, \ksscommitmentrand_b$ for the user
  to obtain a signature $(\bssignature_b, \bsfreshcommitment_b, \bsfreshcommitmentrand_b)$ on $m_b$.
  Next, it runs
  the $\bssign(\bspk, \ksscommitment_{1-b})$ protocol
  with private input $m_{1-b}, \ksscommitmentrand_{1-b}$ for the user
  to obtain a signature $(\bssignature_{1-b}, \bsfreshcommitment_{1-b}, \bsfreshcommitmentrand_{1-b})$ on $m_{1-b}$.
  If both protocols are successful, the challenger sends $(\bssignature_0,
  \bsfreshcommitment_0)$ and $(\bssignature_1, \bsfreshcommitment_1)$ to the
  adversary. Otherwise, it sends nothing.
\item[Guess] Finally, the adversary outputs a guess $\challengebit'$ of $\challengebit$. The adversary wins if $\challengebit' = \challengebit$.
\end{description}
\end{game}
If no adversary can win this game then the signer cannot recognize the
signatures it helped produce.

The computationally hiding commitments in the \obtainkstoken protocol ensure that the \TS learns nothing about the unrevealed ciphertexts $\ctxt$ and witness ciphertexts $\ctxt_i$ which it blindly signs. So, when the user runs $\genshares$ and thereby reveals these ciphertexts, they cannot be directly correlated to a corresponding run of \obtainkstoken. Moreover, the plaintext corresponding to the ciphertexts $\ctxt_i$ are fully randomized, so that these too do not reveal anything about the user with which the \TS is currently interacting.

The privacy proof folows a sequence of games. Throughout we use a guess $i_0, i_1$ for the challenge tokens. If this guess turns out to be incorrect when the adversary makes it challenge query, we abort and try again. We first use a sequence of games to show that we can remove identifying information from the $\obtainkstoken$ protocol.
\begin{itemize}
\item Game 0 is the Tandem privacy game, see page~\pageref{game:privacy}.
\item In Game 1, we extract the \TS key-shares $\sks{0}$ and $\sks{1}$ for users $\User_0$ and $\User_1$ from the \TS' proof of knowledge in step I of the \registeruser protocol, see Appendix~\ref{sec:appendix-register-proof}.
\item In Game 2, we forge the user's zero-knowledge proof of correct
  construction of $\ksscommitment$, the commitment to the epoch, the user's
  private key $\abcsku$, and the randomized ciphertexts, at the end of \obtainkstoken protocol.
\item In Game 3, we use the extractability of $\extcommit{\cdot}{\cdot}$ to forge the user's cut-and-choose proof in the \obtainkstoken protocol, and send random commitments $\ksscommitment_i, \deltacommitment_i$ for $i \not\in \tokendisclose$. However, we honestly construct $\ksscommitment$ as per the protocol.
\item In Game 4, for user $\User_i$ and the challenge token, we set
  $\ctxt = \penc(\sks{i} + \xsencdelta, \prand)$ and
  $\ctxt_i = \penc(\sks{i} + \kssrandomizer_i, \prand_i)$ for $i \not\in \tokendisclose$
  rather than using $\xsenc$. We commit to $\ctxt_i$ for $i \not\in \tokendisclose$ as usual.
  Lemma~\ref{lem:ctxt-pk-correct} shows that with high probability we still follow the protocol correctly.
\item In Game 5, we omit $\sks{i}$ altogether in the construction of the unrevealed $\ctxt_i$, that is, we set:
  \begin{equation}
    \begin{split}
      \ctxt_i &= \penc(\kssrandomizer_i, \prand_i)
    \end{split}
  \end{equation}
  for all $i \not\in \tokendisclose$. Similarly, we set $\ctxt = \penc(\xsencdelta, \prand)$, and use these values to construct $\ksscommitment$. When answering $\qRunTCP$ queries, user $i$ adds $\sks{i}$, which we extract during the \registeruser protocol, to its long-term secret-share $\xp$ to compensate for this change. The size of the randomizers $\kssrandomizer_i$ and $\xsencdelta$ ensures that the \TS cannot detect this change.
\item In Game 6, we replace the user's private key $\abcsku$ in the commitment $\ksscommitment$ by the value 0. Because of the hiding property of the Pedersen commitment $\ksscommitment$ (and the fact that we simulate the proof of correct generation of $\ksscommitment$) ensure the adversary cannot detect this change.
\item Finally, in Game 7, we simulate the correct opening of the commitment $\kssfreshcommitment$ in the first step of the \genshares protocol without knowing the randomizer $\kssfreshcommitmentrand$. Note that $\kssfreshcommitment$ itself still commits to the same values as in Game 6. Because we simulate the proof, the adversary does not notice this change.
\end{itemize}
\newcommand{\finalgame}{Game 7\xspace}

We are now in the situation where the tokens held by user 0 and 1 are exchangeable. We use this to show that no adversary can distinguish situations $\challengebit = 0$ and $\challengebit = 1$. We use a sequence of games to interpolate between the two situations. We start from \finalgame.
\begin{itemize}
\item In Game A, the challenger uses $\challengebit = 0$ but otherwise proceeds as in \finalgame.
\item In Game B, the challenger swaps the signatures of the challenge tokens of users $\User_0$ and $\User_1$. By the blind signature game, the adversary cannot detect this change.
\item In Game C, the challenger also swaps the users $\User_0$ and $\User_1$ in the challenge phase. As a result, it perfectly simulates $\challengebit = 1$ in \finalgame. The privacy property of the threshold cryptographic protocol (with colluding respectively honest SP) ensures that the adversary cannot detect this change.
\end{itemize} 
Since these steps are indistinguishable, no adversary can distinguish the situations $\challengebit = 0$ and $\challengebit = 1$ in \finalgame, and by indistinguishability again, neither can any adversary distinguish these two in the original privacy game.

\begin{proof}[Proof of Theorem~\ref{thm:tandem-privacy}]
Throughout this proof, we use a guess for the challenge tokens $i_0$ and $i_1$ of users $\User_0$ and $\User_1$ respectively. If this guess turns out to be wrong in the challenge step, we abort and try again.

In Game 1, the challenger extracts $\sks{0}$ and $\sks{1}$ for users $\User_0$ and $\User_1$. In particular, the challenger runs the knowledge extractor on the proof of knowledge of the \registeruser protocol, see Equation~\ref{eq:proof-range-xs}, for each of the users. Since the extractor uses rewinding, the adversary does not detect this.

In Game 2, the challenger forges the proof of knowledge of correctness of the commitment $\ksscommitment$ at the end of the \obtainkstoken protocol for the challenge tokens $i_0$ and $i_1$ of users $\User_0$, $\User_1$ respectively. By simulatability, the adversary cannot detect this change.

In Game 3, the challenger extracts the subset $\tokendisclose$ from the commitment $\discloseSetCommitment$ as soon as it receives it. For the two challenge tokens, the challenger (acting as the user) now proceeds as follows. It computes $\ksscommitment_i, \deltacommitment_i$ for $i \in \tokendisclose$ as per the protocol. However, for $i \not\in \tokendisclose$ it lets the unrevealed commitments $\ksscommitment_i$ and $\deltacommitment_i$ commit to random values. The proof of knowledge that $\ksscommitment$ commits to the same values as $\ksscommitment_i$ is already forged since a previous step. Because the commitment scheme is computationally hiding, the adversary cannot detect this change. Despite the changes we made, the final token that is stored by the user is exactly the same as in the original $\obtainkstoken$ protocol.

In Game 4 we compute the values $\ctxt$ and $\ctxt_i$ for user $\User_j$ and $i \not\in \tokendisclose$ as
\begin{equation*}
  \begin{split}
    \ctxt &= \penc(\sks{j} + \xsencdelta; \prand) \\
    \ctxt_i &= \penc(\sks{j} + \kssrandomizer_i; \prand_i)
  \end{split}
\end{equation*}
(recall, we extracted $\sks{j}$ in the \registeruser phase) instead of $\ctxt = \xsenc \cdot \penc(\xsencdelta; \prand)$ and $\ctxt_i = \xsenc \cdot \penc(\kssrandomizer_i; \prand_i)$. Lemma~\ref{lem:ctxt-pk-correct} shows that with overwhelming probability $\pdec(\xsenc)$ equals the value $\sks{j}$ we extracted in the \registeruser protocol, so this change does not modify the adversary's view.

In Game 5 the user omits $\sks{j}$ in the computation of $\ctxt$ and $\ctxt_i$, and instead sets:
  \begin{equation*}
    \begin{split}
      \ctxt &= \penc(\xsencdelta; \prand) \\
      \ctxt_i &= \penc(\kssrandomizer_i, \prand_i)
    \end{split}
  \end{equation*}
for the challenge tokens. To compensate for the fact that $\sks{j}$ is no longer included, the users adds $\sks{j}$ to $\xp$. As a result, the threshold cryptographic protocol still completes as before.

The size of the domain from which the $\kssrandomizer_i$s and $\xsencdelta$ are drawn, ensures that the adversary cannot detect this change when the users uses the token. More formally, the user sends $\ctxt$, $\ctxt_i$s,
$\kssrandomizerDiff_i$s, and $\prandReveal_i$s. However, the last two sets are redundant, they can be computed directly based on $\ctxt$ and the $\ctxt_i$s. As a result, we can focus on $\xsencdelta = \pdec(\ctxt)$ and $\kssrandomizer_i = \pdec(\ctxt_i)$. By the size of the domain of $\xsencdelta$ and the $\kssrandomizer_i$s and the size of $\sks{j}$ the tuples $(\xsencdelta, \kssrandomizer_1, \ldots, \kssrandomizer_{\tokensecpar})$ and $(\sks{j} + \xsencdelta, \sks{j} + \kssrandomizer_1, \ldots, \sks{j} + \kssrandomizer_{\tokensecpar})$ are statistically indistinguishable. As a result no adversary can distinguish Games 4 and 5.

In Game 6, the user omits their private key $\abcsku$ from the commitment $\ksscommitment$ by setting
\begin{equation*}
  \ksscommitment = \commit((0, \epoch, \hash(\ctxt), \hash(\ctxt_{i_1}), \ldots, \hash(\ctxt_{i_{\tokensecpar}})), \ksscommitmentrand).
\end{equation*}
The hiding property of the commitment scheme, the fact that we simulate the proof of correctness of $\ksscommitment$, and the fact that both challenge users are unrevoked, ensures that the adversary cannot detect this change.

In Game 7, we simulate the proof of knowledge in step~\ref{step:genshares:user-prove} of the \genshares protocol for the challenge users. Because the simulation is perfect, the adversary does not notice this change. Note that the commitment $\kssfreshcommitment$ still commits to the correct values.

We now show that no adversary can win \finalgame. We again use a sequence of games, but now interpolate between Game A, where the challenger uses $\challengebit = 0$ in \finalgame, and Game C, where the challenger uses $\challengebit = 1$ in \finalgame. We construct the intermediate Game B, where user $\User_0$ uses the token $i_1$ of user $\User_1$ and vice versa. Since the challenge tokens in \finalgame (and thus in Games A, B, and C) do not depend on the user, the TCPs complete correctly as in \finalgame.

We first show that Games A and B are indistinguishable. Suppose to the contrary that $\Adv$ can distinguish Games A and B. We show that we can use $\Adv$ to build an adversary $\AdvB$ that breaks the blindness property of the signature scheme. In the blind signature game, $\AdvB$ gets oracle access to two users that request a blind signature on one message each. Adversary $\AdvB$ acts as the challenger towards $\Adv$ in \finalgame. At the start of the game $\AdvB$ generates two messages, corresponding to key-share tokens, for which users $\User_0$ and $\User_1$ need a blind signature. It creates:
  \begin{align*}
    m_0 &= (0, \epoch,  \hash(\ctxt_1), \ldots, \hash(\ctxt_{\tokensecpar}) ) \\
    m_1 &= (0, \epoch, \hash(\ctxt_1'), \ldots, \hash(\ctxt_{\tokensecpar}') ),
  \end{align*}
where the values in the tuples are as in \finalgame. Adversary $\AdvB$ sends $m_0, m_1$ to its blind signature challenger.

During the $\obtainkstoken$ protocols for the challenge tokens, $\AdvB$ simulates its users as follows. When user $\User_0$ is running the blind signature protocol to create the challenge token $\ksstoken_0$, $\AdvB$ uses its challenger of the blind signature game to act as the user. When $\User_1$ runs the blind signature protocol to create token $\ksstoken_1$, $\AdvB$ again uses its blind signature game challenger. Finally, the blind signature challenger outputs two signatures $\tokensignature_0$ and $\tokensignature_1$ on messages $m_0$ and $m_1$ respectively. Adversary $\AdvB$ uses $\tokensignature_0$ to construct the key-share token for user $\User_0$, and uses $\tokensignature_1$ to construct the key-share token for user $\User_1$. Note that the blind signature challenger does not output the randomizers for the commitments $\kssfreshcommitment_0$ and $\kssfreshcommitment_1$, but this does not matter as we simulate the proof of knowledge that requires them.

If $\challengebit = 0$ in the blind-signature game, $\AdvB$s challenge user first blindly signed $m_0$, so $\AdvB$ perfectly simulates Game A. If $\challengebit = 1$ in the blind-signature game, then $\AdvB$ perfectly simulates Game B. Hence, any distinguisher between Games A and B breaks the blindness property of the blind signature scheme.

We now show that if the TCP scheme is private (with a colluding respectively honest SP), no adversary can distinguish between Games B and C. Suppose to the contrary that adversary $\Adv$ can distinguish Game B from Game C. We show that we can use $\Adv$ to build an adversary \AdvB that breaks the privacy property of the TCP scheme. Adversary \AdvB simulates users $\User_{0}$ and $\User_{1}$ towards $\Adv$.
The \registeruser and \obtainkstoken protocols do not involve the users' secrets, so \AdvB computes them directly. We now show how to answer $\qRunTCP$ queries.

\comment[WL]{Check MARKER and update token content}
Whenever $\Adv$ makes a $\qRunTCP(\User_i, j, \auxu)$ query, \AdvB makes a $\qRunTCP(i, \auxu)$ query of its challenger. Distinguisher \AdvB's challenger replies with the \TS' key-share $\xsfresh$. Let $\ksstoken = \ksstokencontent[l]$ be the $j$th token of user $\User_i$. Normally, this token dictates a \TS key-share unequal to $\xsfresh$, but we can use the random oracle and change the token to ensure that the \TS will recover $\xsfresh$.
To do so, the adversary sets $\xsencdelta' = \xsencdelta + (\xsfresh - (\xsencdelta \mod \grouporder))$ so that $\xsencdelta' \mod{\grouporder} = \xsfresh$, and then computes $\ctxt = \penc(\xsencdelta'; \prand)$. (Note that the size of $\xsencdelta'$ is correct with overwhelming probability).
Adversary $\AdvB$ updates the random oracle to ensure that $\hash(\ctxt') = \hash(\ctxt)$, i.e., the new ciphertexts has the same hash value as the original pairs. Next, \AdvB uses token $\ksstoken' = (\tokensignature, \kssfreshcommitment, \kssfreshcommitmentrand, \ctxt', \xsencdelta', (\ctxt_l', \prand_l', \kssrandomizer_l')_{l=1,\ldots,\tokensecpar})$ to run $\genshares$ with the \TS.

The changes to the random oracle ensure that this token is valid. Moreover, the changes to the random oracle succeed with high probability since at no point in the games does the \TS learn the inputs to these hash-functions. The \TS will derive the correct secret share $\xsfresh$ from $\ksstoken'$. So it runs the correct TCP protocol with the requested user which is simulated by \AdvB's challenger.

To answer $\Adv$'s challenge queries, \AdvB again uses his challenger and proceeds as above to answer the queries. If $\challengebit = 0$ in the TCP privacy game, then \AdvB's first run of $\qRunTCP$ 
uses user $\User_0$'s key, so \AdvB simulates Game B. Otherwise, if $\challengebit = 1$, then \AdvB simulates Game C. So, any adversary $\Adv$ that can distinguish Games B and C breaks the privacy property of the TCP scheme. This completes the privacy proof.
\end{proof} 
\end{document}